%% file: main.tex
\documentclass[sigconf,nonacm,9pt]{acmart}

\newcommand\vldbdoi{XX.XX/XXX.XX}

\newcommand\vldbavailabilityurl{http://vldb.org/pvldb/format_vol14.html}

\usepackage[english]{babel}
\usepackage{caption}
\usepackage{subcaption}
\usepackage{algpseudocode}
\usepackage{lscape}
\usepackage[linesnumbered,ruled,vlined]{algorithm2e}
\usepackage{amsmath}
\usepackage{graphicx}
\usepackage{cleveref}
\usepackage{array}
\usepackage{booktabs}
\usepackage{mathtools}

\usepackage{algpseudocode}
\usepackage{listings}
\usepackage{color}
\usepackage{courier}

\newcommand{\eps}{\epsilon}

\def\pr{\mathbb{P}}

\def\calD{\mathcal{D}}

\def\calN{\mathcal{N}}

\newtheorem{theo}{Theorem}[section]
\newtheorem{lemma}[theo]{Lemma}

\newcounter{theoremc}
\newtheorem{theor}[theoremc]{Theorem}
\newtheorem{prop}[theoremc]{Proposition}
\newenvironment{hproof}{%
  \renewcommand{\proofname}{Proof Sketch}\proof}{\endproof}

\newcounter{prob}
\newtheorem{problem}[prob]{Problem}
\newtheorem*{problem*}{Problem}

\newcounter{exp}
\newtheorem{example}[exp]{Example}

\newcounter{def}
\newtheorem{definition}[def]{Definition}

\theoremstyle{definition}

\theoremstyle{remark}

\usepackage{dsfont}

\DeclareMathOperator*{\argmax}{argmax}

\DeclareMathOperator*{\Lap}{Lap}

\def\cov{\mathrm{\widehat{\sigma^2_{x y}}}}
\def\var{\mathrm{\widehat{\sigma^2_x}}}
\def\pcov{\mathrm{\widetilde{\sigma^2_{x y}}}}
\def\pvar{\mathrm{\widetilde{\sigma^2_x}}}

\def\HS{\hspace{\fontdimen2\font}}
\definecolor{darkgreen}{rgb}{0.15,0.55,0.15}
\definecolor{darkblue}{rgb}{0.1,0.1,0.5}
\definecolor{blue}{rgb}{0.01,0.40,.8}
\definecolor{darkgreen}{rgb}{0.15,0.55,0.15}
\definecolor{mred}{rgb}{.80,.12,.30}
\definecolor{grey}{rgb}{0.5,0.5,0.5}
\definecolor{Purple}{rgb}{.75,0,.85}
\definecolor{light-gray}{gray}{0.95}
\definecolor{mid-gray}{gray}{0.85}
\definecolor{darkred}{rgb}{0.7,0.25,0.25}
\definecolor{rose}{rgb}{1.0, 0.01, 0.24}
\newcommand{\red}[1]{\textcolor{red}{#1}}

\newcommand{\blue}[1]{\textcolor{blue}{#1}}

\newcommand{\revise}[1]{#1}

\newcommand{\eat}[1]{}

\newcommand{\stitle}[1]{\vspace{2pt}\noindent\textbf{#1}}

\newcommand{\fdp}[0]{\texttt{FPM}\xspace}
\newcommand{\fdpopt}[0]{\texttt{FPM-OPT}\xspace}
\newcommand{\ldp}[0]{\texttt{LDP}\xspace}
\newcommand{\gdp}[0]{\texttt{GDP}\xspace}
\newcommand{\sdp}[0]{\texttt{SDP}\xspace}
\newcommand{\sdpone}[0]{\texttt{Shuffle-1}\xspace}
\newcommand{\sdptwo}[0]{\texttt{Shuffle-2}\xspace}
\newcommand{\sys}[0]{{\it Saibot}\xspace}

\newenvironment{myitemize}{\begin{list}{$\bullet$}{\renewcommand{\leftmargin}{0.15in}}}{\end{list}}

\author{Zezhou Huang}
\email{zh2408@columbia.edu}
\affiliation{
  \institution{Columbia University}
}

\author{Jiaxiang Liu}
\email{jl6235@columbia.edu}
\affiliation{
  \institution{Columbia University}
}

\author{Daniel Gbenga Alabi}
\email{alabid@cs.columbia.edu}
\affiliation{
  \institution{Columbia University}
}

\author{Raul Castro Fernandez}
\email{raulcf@uchicago.edu}
\affiliation{
  \institution{University of Chicago}
}

\author{Eugene Wu}
\email{ewu@cs.columbia.edu}
\affiliation{
  \institution{DSI, Columbia University}
}

\begin{document}

\title{Saibot: A Differentially Private Data Search Platform}

\begin{abstract}
\input{sections/abstract}
\end{abstract}

\maketitle



\input{sections/intro.tex}

\input{sections/overview.tex}

\input{sections/solution.tex}

\input{sections/analysis.tex}

\input{sections/experiments.tex}

\input{sections/related.tex}
\input{sections/conclusions.tex}

\pagebreak

\bibliographystyle{ACM-Reference-Format}
\bibliography{main}

\clearpage
\appendix
\input{sections/fdpproof}
\input{sections/erroranalysis}
\input{sections/unbiasedproof.tex}
\input{sections/linear_regression_proof.tex}
\input{sections/noise_allocation}

\end{document}

%% file: sections/abstract.tex
Recent data search platforms use ML task-based utility measures rather than metadata-based keywords, to search large dataset corpora.   
Requesters submit a training dataset, and these platforms search for {\it augmentations}---join or union-compatible datasets---that, when used to augment the requester's dataset, most improve model (e.g., linear regression) performance.
Although effective, providers that manage personally identifiable data demand differential privacy (\texttt{DP}) guarantees before granting these platforms data access.   Unfortunately, making data search differentially private is nontrivial, as a single search can involve training and evaluating datasets hundreds or thousands of times, quickly depleting privacy budgets.  

We present \sys, a differentially private data search platform that employs  Factorized Privacy Mechanism (\fdp), a novel  \texttt{DP} mechanism, to calculate sufficient semi-ring statistics for ML over different combinations of datasets. These statistics are privatized once, and can be freely reused for the search.
This allows \sys to scale to arbitrary numbers of datasets and requests, while minimizing the amount that \texttt{DP} noise affects search results.  
We optimize the sensitivity of \fdp for common augmentation operations, and analyze its properties with respect to linear regression.  Specifically, we develop an unbiased estimator for many-to-many joins, prove its bounds, and develop an optimization to redistribute \texttt{DP} noise to minimize the impact on the model.
Our evaluation on a real-world dataset corpus of $329$ datasets demonstrates that \sys can return augmentations that achieve model accuracy within $50{-}90\%$ of non-private search, while the leading alternative \texttt{DP} mechanisms (\texttt{TPM}, \texttt{APM}, shuffling) are several orders of magnitude worse.

%% file: sections/intro.tex
\section{Introduction}
\label{sec:intro}

Augmenting training data with additional samples or features can significantly enhance ML performance~\cite{sambasivan2021everyone}. However, sourcing such data in large corpora---public portals~\cite{nycopen,cms}, or enterprise data warehouses---is a complex task.
To address this, a new form of data search platform~\cite{santos2022sketch,chepurko2020arda,nargesian2022responsible,li2021data,kitana} is emerging, wherein a requester submits a search request comprising training and testing datasets for augmentation. The platform then finds provider datasets that augment the training dataset in a way that improves {\it utility} (e.g., ML performance).  This involves using a data discovery tool~\cite{fernandez2018aurum,castelo2021auctus} to locate a set of union- or join-compatible tables ({\it augmentations}), augmenting the training set with each candidate, and then retraining and evaluating the model to assess its {\it utility}.  The augmentations are subsequently ranked by {\it utility}.  Platforms largely differ in the discovery tool procedure, the models they support, and how they accelerate model retraining and evaluation.   Recent works~\cite{kitana,chen2017semi,santos2022sketch} suggest that using linear regression as a model proxy provides a good balance of search quality and runtime.

Unfortunately, privacy is a major barrier to sharing for many potential data providers and requesters with sensitive data (e.g., personally identifiable information (PII), and protected health information (PHI)).     In these cases, providers are legally obligated to prevent personal data leakage~\cite{EUdataregulations2018,CCPA,FERPA}.  Rather than prohibit access outright, differential privacy (\texttt{DP})~\cite{dwork2006calibrating} supports data analysis on sensitive data while bounding the degree of privacy loss based on the budget $\epsilon$ set by the data provider.  Each query on the dataset adds noise to the results, inversely proportional to the budget consumed; when $\epsilon=0$, the dataset becomes inaccessible.

Ideally, a differentially private data search platform would let providers and requesters set privacy budgets for their datasets, and enforce these budgets as new datasets and requests arrive.  
Moreover, since the platform is often a third-party service that may not be trusted by data providers (and the individuals they collect data from), it should not have access to raw data. 
Unfortunately, integrating \texttt{DP} with data search platforms is non-trivial.
To illustrate, \Cref{fig:dp_illustration} shows where existing mechanisms
would add noise in a two-level data-sharing architecture that matches many real-world settings.   In this architecture,  individuals (e.g., patients) generate sensitive data aggregated by providers/requesters (e.g., hospitals), and the search platform further aggregates their datasets.   

\revise{Global \texttt{DP} (\gdp) is a \texttt{DP} definition widely used by private DBMSes~\cite{johnson2018towards, wilson2019differentially, kotsogiannis2019privatesql}, where the employed mechanisms add noise after executing, e.g., a query over private data by a trusted central DBMS. However, when applied to data search, previous \gdp mechanisms} need to ``split the budget'' across every candidate augmentation on every request.  The budget ends up being so small that the noise drowns any signal in the data.  \revise{Further, their trust model requires the search platform, acting as the central aggregator, to be trusted, which is challenging since it is a third-party service. 
To address this, mechanisms for Local \texttt{DP} (\ldp) (e.g., randomized response~\cite{erlingsson2014rappor,ding2017collecting}) eliminate the need for a trusted data curator by privatizing} 
individual tuples. Nevertheless,
the noise required for these mechanisms can be quite large, potentially compromising data utility~\cite{wei2020federated}.
\revise{Shuffling~\cite{erlingsson2019amplification,feldman2022hiding} is a mechanism for an intermediate trust model that, instead of relying on a trusted central aggregator, requires trust in a shuffler. After privatizing tuples (using mechanisms for \ldp), the shuffler shuffles the primary keys of tuples during aggregation to disassociate them from individuals; this "amplifies privacy" by allowing each tuple to have less noise applied. Variation \sdpone shuffles at the provider/requester level but requires considerable noise for small datasets; \sdptwo shuffles within the search platform but needs to trust the platform.}
An alternative to shuffling, widely used by federated ML~\cite{wei2020federated,shokri2015privacy,truex2020ldp,zhao2020local}, is to let providers/requesters iteratively compute and privatize model gradients locally, and let an untrusted aggregator compute the final model.   However, these gradients are specific to a single augmentation's model, so the budget is still split across all candidate augmentations.
 
\input{sections/introfig.tex}

Is it possible for a DP search platform to return search results of comparable quality to non-private search, {\it and} for the platform to scale to many datasets and requests?  
We are motivated by the recent data search platform Kitana~\cite{kitana}, which uses semi-ring aggregation to quickly evaluate a candidate augmentation's {\it utility} on a linear regression model without materializing the augmented table and fully retraining it.  
These semi-rings can be computed for each dataset offline, and Kitana only needs these semi-rings to evaluate a candidate augmentation in ${\approx}1{-}5ms$, independent of the dataset size. {\it Our main observation is that these precomputed semi-rings also serve as ideal intermediates for \texttt{DP}, as they help directly estimate model parameters, can be combined over joins and unions, and can be freely reused once made private.}

This paper presents \sys, a differentially private data search platform for tabular datasets that scales to unlimited datasets and requests, returns results comparable to non-private search, and doesn't need to be trusted. 
Data providers upload their privatized datasets to the platform.   When a requester submits a privatized training dataset, the platform searches for the best combinations of privatized datasets which, when augmented with the requester's dataset, most improve the accuracy of a linear regression model. 
\revise{For the trust model, \sys assumes that the $1^{st}$-level aggregators are trusted (unlike the local model) but the $2^{nd}$-level aggregators (i.e., search platform) are not (unlike global model).}
In practice, regulations~\cite{EUdataregulations2018,CCPA,HIPAA,FERPA} mandate that the $1^{st}$-level aggregators (e.g., healthcare providers, schools) securely store individual data.
\revise{
Once \sys identifies predictive augmentations using differentially private proxy models (linear regression), it can directly return the private proxy models, although they may not be complex enough for some requesters. To address this, \sys can be integrated within a larger differentially private federated ML system~\cite{wei2020federated,truex2020ldp,zhao2020local,wang2020hybrid} to train more advanced models, like deep neural networks through differentially private gradient descent, on the identified augmentations.}.

Our key innovation is a new \texttt{DP} mechanism called \revise{{\it Factorized Privacy Mechanism} (\fdp)}, where each requester or provider computes and privatizes sufficient statistics on their own datasets based on their privacy requirements.  These sufficient statistics provide high utility, can be freely reused for ML over different augmentations, and only require the search platform to store privatized datasets. \revise{\fdp satisfies \gdp, but the randomized algorithm is applied by the
$1^{st}$-level aggregators rather than the $2^{nd}$-level ones. Note that \fdp has broader applications, not only for the data search but also for more general differentially private factorized learning.}

\revise{{\bf The main algorithmic challenge \fdp solves is to design privatized sufficient statistics for ML that are composable to support various join and union augmentations.}}
Previous works have applied \texttt{DP} to sufficient statistics for privatized linear regression~\cite{wang2018revisiting} and GLM~\cite{huggins2017pass,kulkarni2021differentially}, but these sufficient statistics can be used for only a single dataset. 
Our key insight is to design these sufficient statistics as a semi-ring~\cite{green2007provenance}, which includes addition and multiplication operators for union and join. Although  sufficient semi-ring statistics have been utilized for ML~\cite{schleich2019layered,schleich2016learning} over joins, we are the first to explore their application in a \texttt{DP} setting. The results of our real-world experiments indicate that \fdp is capable of identifying augmentations that achieve an average $r2$ score of ${\sim}50{-}90\%$ compared to non-private searches. Additionally, \fdp can support a large data corpus and unlimited requests. In contrast, the other baseline mechanisms achieve $r2$ scores ${<}0.02$.

To summarize, our contributions are as follows: 

\begin{myitemize}
\itemsep0em 
  \item \revise{We propose \fdp, a novel \texttt{DP} mechanism that privatizes reusable and composable monomials for join/aggregation augmentations. We integrate \fdp into \sys to achieve scalability for large volumes of datasets and search requests with high utility.}
  \item We optimize \fdp based on the parity of the statistics order.  For the special case of tables containing a single feature, we reduce the expected error by a further factor of $\sqrt{2}$.

  \item We provide a deep analysis of \fdp to linear regression models.   Specifically, we study the statistical bias introduced in many-to-many joins, and design an unbiased estimator to address this.   We further study its bounds on errors over the model parameters.  
    \item We design an optimization that carefully redistributes noise across sufficient statistics to improve linear regression accuracy.  
  \item We thoroughly evaluate \fdp across a real-world data corpus with ${>}300$ datasets. Our results show that \fdp can accurately identify augmentations that achieve $r2$ scores close to (${\sim} 50 {-} 90\%$) those of a non-private search. We further use ablation studies to validate our theoretical analyses and study the sensitivity.
\end{myitemize}

\stitle{Note:} The paper is self-contained. References to appendices can be disregarded or located in the technical report~\cite{tech}.

%% file: sections/introfig.tex
\begin{figure}
  \centering
      \includegraphics [width=0.45\textwidth]  {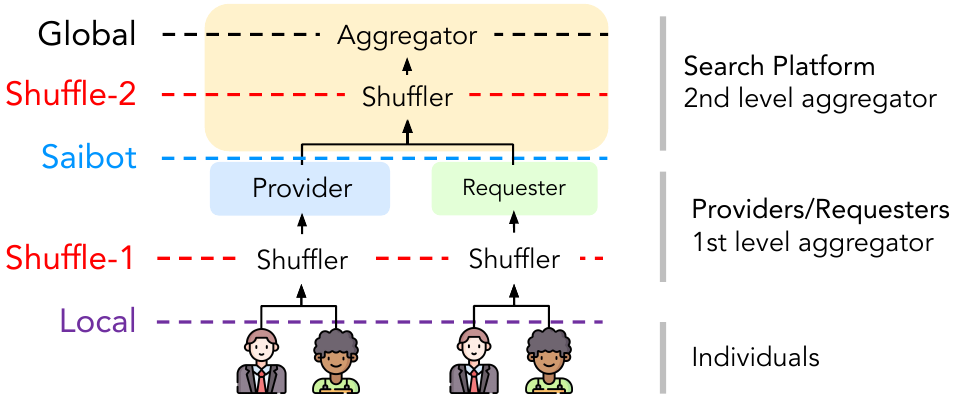}
        \vspace*{-4mm}
  \caption{\revise{Summary of DP-mechanisms under various trust models in a standard data sharing architecture: Providers/requesters collect data from individuals, and the search platform aggregates data from providers/requesters. At the extremes, mechanisms for local models introduce noise to individual tuples, whereas naive mechanisms for global models add noise to query results through a trusted $2^{nd}$-level aggregator. Mechanisms for the shuffle model introduce a shuffler at a point of aggregation (either at the $1^{st}$-level (\sdpone) or the $2^{nd}$-level (\sdptwo)). In contrast, \sys adds noise to sufficient statistics computed by providers/requesters.}}
   \vspace*{-8mm}
  \label{fig:dp_illustration}
\end{figure}

%% file: sections/overview.tex
\section{Private Task-based Data Search}

In this section, we formalize the problem of task-based private data search. We start with an introduction of the non-private problem and current solutions. We then provide the primer of differential privacy, and present the differentially private data search problem.
 
\subsection{Non-Private Task-based Data Search}

We provide the background of previous task-based data search problem~\cite{santos2022sketch,chepurko2020arda,nargesian2022responsible,li2021data}, which is non-private, and previous solutions. 

\stitle{Data Model.}  
We follow the standard relational data model.  Relations are denoted as $R$, attributes as $A$, and domains as $dom(A)$. $R$'s schema is represented by $S_R=[A_1,\cdots,A_n]$, with tuples labeled as $t$ and attribute values as $t[A]$. For clarity, the schema is included in square brackets following the relation in examples $R[A_1,\cdots,A_n]$. The domain of a relation is the Cartesian product of attribute domains: $dom(R) = dom(A_1)\times\cdots\times dom(A_n)$. We consider each dataset as a relational table and use these terms interchangeably.

\stitle{Machine Learning.} A ML task $M$, like linear or logistic regression, aims to fit a good model based on feature-target attribute pairs. A training dataset $R_{train}$ comprises features $X\subset S_R$ and a target attribute $y\in S_R$. The task $M$ has a training function $M.Train(\cdot)$ that inputs $R_{train}$ and outputs a model $m$ that optimally predicts $y$ from $X$, even for unseen $X, y$ pairs. To assess $m$, $M$ uses a function $M.Evaluate(\cdot)$ which inputs $m$ and a testing dataset $R_{test}$, and outputs the model's performance on $R_{test}$, typically measured by accuracy, which is to be maximized.

\stitle{Task-based Data Search.} Given a data corpus with datasets from different providers, requesters send a request with datasets to augment and a task (e.g., ML). Task-based data search aims to identify a set of augmentable (join/union) datasets that maximize task utility.

To formalize this, let $\mathcal{R} = \{R_1, R_2,...\}$ be a data corpus with a set of relations, with each from some provider. Requester sends a request with training and testing dataset $(R_{train},R_{test})$, and chooses a model $M$. Requester's goal is to train model $M$  on $R_{train}$ and maximize its performance on $R_{test}$, which we call the task's {\it utility}.

To improve the {\it utility}, the requester aims to find a set of provider datasets in $\mathcal{R}$ that can be used to augment their data and enhance model performance. The function $Discover(R, augType)$ is used to find datasets in the data corpus $\mathcal{R}$ that can be joined or unioned with $R$, given $augType \in \{\Join, \cup\}$. The requester wants to try different combinations of subsets of these datasets to augment\footnote{For simplicity, we consider datasets that can be directly joined or unioned with requester $R_{train}$. The search space could be further expanded by, e.g., $1^{st}$ joining provider datasets; our solution can be easily adapted to this larger search space.} and find the combination that maximizes {\it utility}.

Putting everything together, the problem can be formulated as:
\noindent\begin{problem}[Task-Based Data Search.]
\label{searchprob}
For request $(R_{train},R_{test},M)$, find the set of datasets $\mathbf{R}^*_\cup, \mathbf{R}^*_\Join  \subseteq \mathcal{R}$ from data corpus such that
\begin{flalign*}
\mathbf{R}^*_\cup, \mathbf{R}^*_\Join=&\quad\argmax_{\mathbf{R}_\cup, \mathbf{R}_\Join}   M.Evaluate(m, R_{testAug})\\
s.t. &\quad \mathbf{R}_\cup \subseteq Discover(R, \cup), \mathbf{R}_\Join \subseteq Discover(R, \Join),\\
&\quad R_{trainAug} = (R_{train} \cup_{R_1\in\mathbf{R}_\cup}  R_1)\Join_{R_2\in\mathbf{R}_\Join} R_2\\
&\quad R_{testAug} = R_{test} \Join_{R\in\mathbf{R}_\Join} R\\
&\quad m = M.Train(R_{trainAug})\\
\end{flalign*}
\end{problem}
\vspace*{-6mm}
\stitle{Solutions.} 
Current task-based data search platforms~\cite{santos2022sketch,chepurko2020arda,nargesian2022responsible,li2021data} follow the architecture illustrated in black in \Cref{fig:arch}. 
Offline, when providers upload raw datasets to {\it Data storage}, the platform computes minhashes for data discovery~\cite{castelo2021auctus,fernandez2018aurum}, and sketches to accelerate retraining~\cite{santos2022sketch,chepurko2020arda,nargesian2022responsible,kitana}. 
Online, the platform solves \Cref{searchprob} for each request  $(R_{train},R_{test},M)$.   First, {\it data discovery}~\cite{fernandez2018aurum,castelo2021auctus} uses the minhashes or sketches to return a set of candidate datasets. {\it Data search} then identifies a subset that maximizes task utility. The brute-force search evaluates all possible combinations and can be expensive due to retraining costs and the large set of combinations, so approaches use various heuristics and greedy algorithms~\cite{santos2022sketch,chepurko2020arda,li2021data}.

Our work primarily builds on Kitana~\cite{kitana}, which follows the architecture in \Cref{fig:arch} and uses specialized sketches for factorized ML. factorized ML trains models over joins without materializing them, which speeds up model retraining and evaluation after any candidate augmentation. This allows Kitana to execute task-based searches much faster, while maintaining competitive task utility.  Our insight is that these sketches boost performance and act as the ideal sufficient statistics for \texttt{DP}, as detailed in \Cref{sec:suffstat}.

\subsection{Differential Privacy Primer}
\label{sec:dpback}
Before delving into our solution to differentially private dataset search, we first introduce differential privacy (\texttt{DP}).  We focus on the Gaussian mechanism, a common, straightforward technique offering comparable performance and guarantee with other baselines (e.g., it offers the same approximate \texttt{DP} by shuffling~\cite{erlingsson2019amplification}).
In practice, our solution can also support pure \texttt{DP}  by Laplace mechanism (\Cref{sec:abalationexp}), where shuffling falls short.

\stitle{Differential Privacy.} \texttt{DP}~\cite{dwork2006calibrating} is a technique used to protect reconstruction, membership, and inference attacks~\cite{dwork2017exposed} by bounding the information leakage from individual records. \texttt{DP} guarantees that the probability that an algorithm will produce the same output on two datasets that differ by only one record is bounded. Formally:

\begin{definition}[$(\epsilon,\delta)-DP$]
Let $f$ be a randomized algorithm that takes a relation $R$ as input. $f$ is $(\epsilon, \delta)-DP$ if, for all relations $R_1, R_2$ that differ by adding or removing a row, and for every set $S$ of outputs from $f$, the following holds:
$Pr[f(R_1) \in S] \leq e^{\epsilon} Pr[f(R_2) \in S] + \delta$,
where $\epsilon$ and $\delta$ are non-negative real numbers (called privacy budget). $\epsilon$ controls the level of privacy, and $\delta$ controls the level of approximation. For the special case when $\delta=0$, $(\epsilon,0)-DP$ is also called pure \texttt{DP}. 
\end{definition}

\revise{\texttt{DP} definitions can be global (\gdp) or local (\ldp) depending on inputs:
\gdp applies to randomized algorithms that process an entire relation (as an aggregator) described above. In contrast, \ldp guarantees the differential privacy of algorithms on individual tuples (or relations with a cardinality of 1) before transmitting tuples to any aggregator. As a result, \ldp algorithms can function under a weaker trust model, where no aggregator is trusted. However, this often leads to increased noise levels and reduced data utility~\cite{yang2020local}.}

There are three important theorems of \texttt{DP}:

\begin{theor}[Robustness to Post-Processing]
\label{theo:post}
Let $f$ be a randomized algorithm that provides $(\epsilon, \delta)-DP$. Let $g$ be an arbitrary function. Then, the composition $g \circ f$ provides $(\epsilon, \delta)-DP$.
\end{theor}

\begin{theor}[Sequential Composition]
Let $f_1,\dots,f_n$ be a sequence of independent algorithms that provide $(\epsilon_1, \delta_1), \dots, (\epsilon_n, \delta_n)-DP$, respectively. Then, the algorithm that applies each of them in sequence, i.e., $f_n \circ f_{n-1} \dots \circ f_1$, is $(\sum_{i=1}^n \epsilon_i, \sum_{i=1}^n \delta_i)-DP$.
\end{theor}

\begin{theor}[Parallel Composition]
Let $dom_1, \dots, dom_n$ be $n$ disjoint subsets of $dom(R)$. Let $f_1,\dots,f_n$ be a set of independent algorithms that provide $(\epsilon_1, \delta_1), \dots, (\epsilon_n, \delta_n)-DP$  and take relations from $dom_1, \dots, dom_n$ as input, respectively. Then, the algorithm that applies them on disjoint subsets of $R$ is $(\max_{i=1}^n \epsilon_i, \max_{i=1}^n \delta_i)-DP$.
\end{theor}

To ensure $(\epsilon,\delta)-DP$ when $\mathcal{Q}$ queries need to be executed, the privacy budget $(\epsilon,\delta)$ can be split among the queries using sequential composition, such as allocating $(\epsilon/\mathcal{Q},\delta/\mathcal{Q})$ for each query. This work employs (basic) sequential composition for simplicity, but it could be further optimized by advanced composition~\cite{dwork2010boosting}.

\stitle{Gaussian Mechanism.} The Gaussian mechanism~\cite{dwork2006our} adds noise to a query function to satisfy $(\epsilon, \delta)$-differential privacy.  Formally:
\vspace*{-2mm}
\begin{theor}[Gaussian Mechanism.]
Given $\eps, \delta \in (0, 1]$, let query $q$ be a function that takes $R$ as input and outputs a vector of real numbers. The Gaussian mechanism independently adds random noise to each output to satisfy $(\epsilon, \delta)$-differential privacy:
$q'(R) = q(R) + \mathcal{N}(0, \sigma^2)$,
where $\mathcal{N}(0, \sigma^2)$ denotes a Gaussian distribution with mean 0 and standard deviation $\sigma = \sqrt{2\ln(1.25/\delta)}\Delta_ q/\epsilon$. $\Delta_q$ is the $l_2$-sensitivity of $q$ defined as: for all possible neighbouring relations $R_1, R_2$, $\Delta_q$ is the maximum $\ell_2$ distance of $q$ outputs $\|q(R_1) -q(R_2) \|_2$.
\end{theor}
\vspace*{-2mm}
Different definitions exist for neighbouring relations (and can be extended to multi-relations). We adopt bounded \texttt{DP}~\cite{Alabi22}, where neighbouring relations $R_1, R_2$ have identical row numbers, but one row's data differ; our system can be readily adapted for other definitions (e.g., unbounded \texttt{DP} where row numbers differ).

\subsection{Private Task-based Data Search}
\label{sec:problem}

\begin{figure}
  \centering
      \includegraphics [width=0.4\textwidth]  {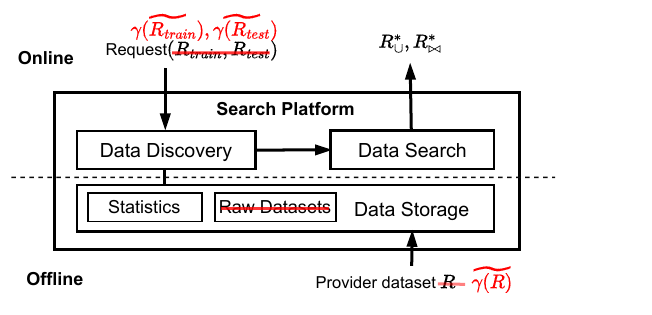}
      \vspace*{-5mm}
  \caption{\sys architecture. 
  Previous data search platforms (black) store raw datasets in data storage, use data discovery to identify augmentable datasets, and search the datasets for task improvement. 
  To ensure privacy, \sys additionally applies \fdp (\red{red}) to compute sufficient semi-ring statistics, that are aggregated ($\gamma$) and privatized ($\sim$) before being sent to the search platform. These statistics can support join and union queries to train and evaluate ML as post-processing.}
  \vspace*{-4mm}
  \label{fig:arch}
\end{figure}

We first lay out the privacy requirements based on the criteria (\Cref{sec:intro}) and motivated by real-world use cases. Then, we define the differentially private data search problem, and discuss the challenges and the intuition for solutions.

\stitle{\revise{Trust Model.}} 
We adopt a standard two-level aggregator setting illustrated in \Cref{fig:trust_model}: the $1^{st}$-level aggregators are providers/requesters (e.g., hospitals, schools), and the $2^{nd}$-level aggregator is the search platform.
Individuals share data with their direct $1^{st}$-level aggregator, who is trusted (e.g., a  hospital collects data from patients and stores them securely). However, they don't trust other non-direct $1^{st}$-level aggregators or the $2^{nd}$-level aggregators (e.g., patients don't trust other hospitals and the search platform).

\revise{Our trust model sits between the global model (by \gdp) and local model (by \ldp): Previous global model~\cite{johnson2018towards, wilson2019differentially, kotsogiannis2019privatesql} assumes that the central data curator ($2^{nd}$-level aggregator) is trusted. On the contrary, the local model assumes no trusted aggregators. In contrast to the shuffle model~\cite{erlingsson2019amplification,feldman2022hiding} which requires a trusted shuffler at either $1^{st}$-(\sdpone, similar to ours) or $2^{nd}$-level (\sdptwo, similar to the global model), we don't rely on any trusted shuffler.} In practice, we believe our trust model fits the structure of many organizations, where individuals solely trust their immediate data aggregator (like a hospital or service provider), but do not trust any other aggregators. Further, regulatory requirements~\cite{HIPAA,CCPA} place privacy protection requirements on the $1^{st}$-level aggregator.   

\stitle{Privacy Requirement.}
Providers and requesters hope to disclose datasets to the malicious search platform for augmentation. Each provider or requester sets a \texttt{DP} budget $(\epsilon, \delta)$ for each of their datasets, which is independent of other datasets and the search platform.
As per previous works~\cite{near2021differential,wu2017achieving}, we assume that each individual contributes to exactly one row of one dataset.
In line with prior studies~\cite{wilson2019differentially,kotsogiannis2019privatesql,johnson2018towards},  
 we assume that the schemas and the domains of join keys (as group-by attributes) are public.  

The {\bf differentially private task-based data search} problem is then defined as \Cref{searchprob}, adhering to the above trust model and satisfying the privacy requirements.

\begin{figure}
  \centering
      \includegraphics [width=0.4\textwidth]  {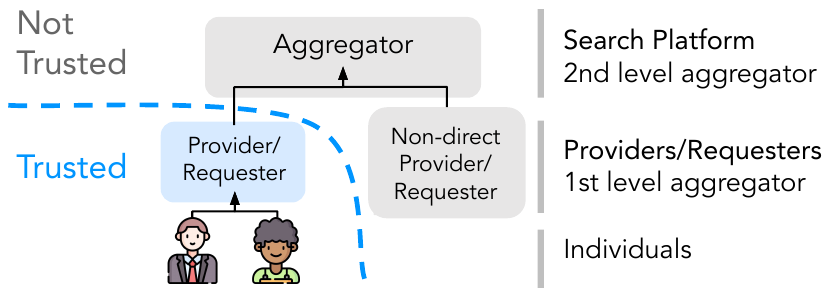}
      \vspace*{-4mm}
  \caption{\revise{Illustration of \sys trust model, where individuals only trust the  direct $1^{st}$-level aggregator, and not any others. }}
   \vspace*{-6mm}
  \label{fig:trust_model}
\end{figure}

\vspace*{-2mm}
\begin{example}
\label{example1}
\revise{Fitbit~\cite{fitbit}, a mobile health app, gathers health data from individuals and is trusted by individuals to handle sensitive information responsibly. To enhance the accuracy of its ML recommender, Fitbit plans to share data with a search platform (as requesters) but also wants to protect sensitive health data. Upon obtaining consent from individuals, Fitbit employs \texttt{DP} to privatize each dataset and uses \sys to search for valuable augmentations.
}
\end{example}
\vspace*{-2mm}

Private task-based data search is particularly challenging because, even for a single request, it requires model retraining over a combinatorially large space of augmented datasets created by joining and unioning candidate datasets. How to avoid exhausting the requester's and the providers' privacy budgets?   How can massive datasets and requests be scaled without degrading search quality?   Is there a one-time differentially private, yet universally useful intermediate representation~\cite{hardt2010multiplicative,blum2013learning}?

We draw inspiration from Kitana~\cite{kitana} which uses factorized linear regression to expedite data search.
Kitana computes the gram matrix semi-ring (\Cref{s:backgroundmsgpassing}) for each dataset, allowing fast join/union with a candidate dataset and evaluation of the linear regression accuracy. While semi-rings were initially used for performance, they also make an ideal intermediate representation for \texttt{DP}. Thus, in the next section, we design \fdp to privatize sufficient semi-ring statistics to support private ML over joins and unions.

%% file: sections/solution.tex
\section{Factorized Privacy Mechanism}

In this section, we introduce Factorized Privacy Mechanism (\fdp), which privatizes sufficient semi-ring statistics. 
We start with the factorized ML background, extend it to monomial semi-ring,  present our main mechanism algorithms, and analyze its errors.

\subsection{Factorized Machine Learning Primer} \label{s:backgroundmsgpassing}

We start with the fundamental concepts of annotated relations and aggregation pushdown, then introduce  factorized ML~\cite{abo2016faq,olteanu2015size}.

\stitle{Annotated Relations.} The annotated relational model~\cite{green2007provenance} maps $t\in R$ to a commutative semi-ring $(D, +, \times, 0, 1)$, where $D$ is  a set, $+$ and $\times$ are commutative binary operators closed over $D$, and $0/1$ are zero/unit elements. An annotation for $t \in R$ is denoted as $R(t)$. Semi-ring annotation expresses various aggregations. For example, the natural numbers semi-ring expresses count aggregations.

\stitle{Semi-ring Aggregation Query.} Semi-ring aggregation queries can now be reformulated using annotated relations by translating group-by, union, and join operations into addition ($+$) and multiplication ($\times$) operations over the semi-ring annotations, respectively.
\begin{align*}
  (\gamma_\mathbf{A} R)(t) = & \sum \{R(t_1) | \HS t_ 1 \in R , t = \pi_{\mathbf{A}} (t_1 )\} \\
  (R_1\cup R_2)(t) =& \HS R_1(t) + R_2(t)  \\
(R_1\Join R_2)(t) =& \HS R_1(\pi_{S_{R_1}} (t)) \times R_2(\pi_{S_{R_2}} (t)) 
\end{align*}
\noindent (1) The annotation for group-by $\gamma_\mathbf{A} R$ is the sum of the annotations within the group. (2) The annotation for  union $R_1 \cup R_2$ is the sum of annotations in  $R_1$ and $R_2$. (3) The annotation for  join $R_1 \Join R_2$ is the product of annotations from contributing tuples in $R_1$ and $R_2$.

\stitle{Aggregation Pushdown.}
The optimization of factorized ML~\cite{abo2016faq,schleich2016learning} involves the distribution of aggregations $\gamma$ (additions) through joins $\Join$ (multiplications). For example, consider the query $\gamma_D (R_1[A,B] \Join R_2[B,C] \Join R_3[C,D])$. Rather than  applying $\gamma$ on the join (which is $O(n^3)$ where $n$ is relation size), $\gamma$ can be performed on $R$ before $\Join$ with $S$, and this process can be repeated two more times (in $O(n)$):
$$\gamma_{D} (\gamma_{C} (\gamma_{B} (R_1[A,B])\Join R_2[B,C]) \Join R_3[C,D])$$
The associativity of additions can be similarly exploited for union: 
$$\gamma_A (R_1[A,B] \cup R_2[A,B]) = \gamma_A (R_1[A,B]) \cup \gamma_A (R_2[A,B])$$

\stitle{Factorized Linear Regression.} The fundamental optimization of factorized ML is aggregation pushdown, but different semi-rings are used for different models. We use linear regression as an example.

We start with an overview of linear regression and its sufficient statistics. Given the training data  $\mathbf{X} \in \mathbb{R}^{n \times m}$, and the target variable $\textbf{y} \in \mathbb{R}^{n\times 1}$, the goal is to find parameters $\theta \in \mathbb{R}^{m \times 1}$ that minimize the square loss $\theta^* = argmin_{\theta}\Vert \textbf{y} - \mathbf{X} \theta\Vert^2$, yielding a closed-form solution $\theta^*{=}(\mathbf{X}^T\mathbf{X})^{-1}\mathbf{X}^T\textbf{y}$. Including the target variable as a special feature and appending it to $\mathbf{X}$ for $\mathbf{X}' {=} [\mathbf{X} \ | \ \textbf{y}]$, we find that $\mathbf{X}'^T\mathbf{X}'{\in}\mathbb{R}^{m' \times m'}$, where $m'=m+1$, is the core sufficient statistics to compute, where each cell represents the sum of products between feature pairs.

We can compute $\mathbf{X}'^T\mathbf{X}'$ over the join $R_{\Join} = R_1 \Join ... \Join R_k$ by the covariance matrix semi-ring~\cite{schleich2016learning}.  For $m'$ features, the semi-ring is defined as a triple $(c,\mathbf{s},\mathbf{Q})\in(\mathbb{Z},\mathbb{R}^{m'}, \mathbb{R}^{m' \times m'})$, which contains the count, sums, and sums of pairwise products respectively. The zero and one elements are $\mathbf{0} = (0,\mathbf{0}^m, \mathbf{0}^{m' \times m'})$ and $\mathbf{1} = (1,\mathbf{0}^{m'}, \mathbf{0}^{m' \times m'})$.  $+$ and $\times$ between two annotations $a$ and $b$ are defined as:
\begin{align*}
a + b =&  (c_a + c_b,\mathbf{s}_a + \mathbf{s}_b, \mathbf{Q}_a + \mathbf{Q}_b) \\
a \times b =& (c_a c_b,c_b\mathbf{s}_a + c_a\mathbf{s}_b, c_b\mathbf{Q}_a + c_a\mathbf{Q}_b + \mathbf{s}_a \mathbf{s}_b^T +\mathbf{s}_b \mathbf{s}_a^T) 
\end{align*}
Then, computing $\mathbf{X}'^T\mathbf{X}'$ is reduced to executing $\gamma(R_1 \Join ... \Join R_k)$, where aggregation can be pushed down as discussed before.

\begin{figure}
  \centering
      \includegraphics [width=0.4\textwidth]  {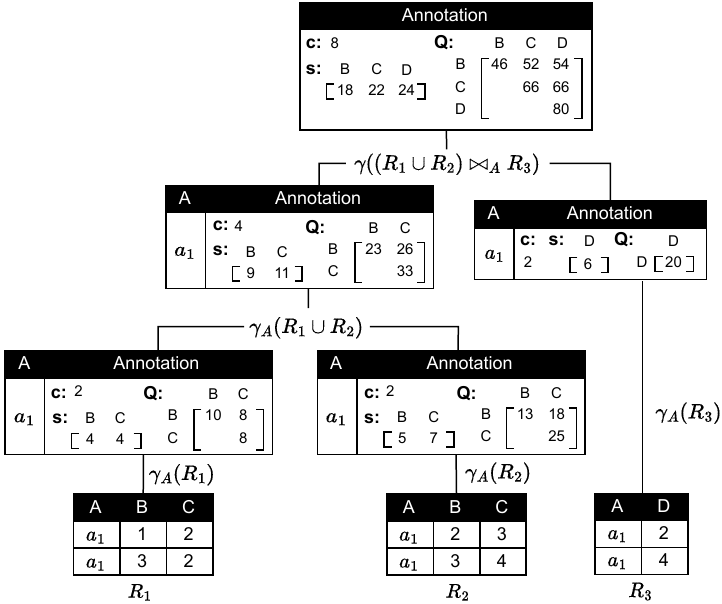}
      \vspace*{-3mm}
  \caption{Optimized query plan of $\gamma((R_1\cup R_2)\Join_A R_3)$ for factorized ML. Aggregations are pushed down before joins.}
  \vspace*{-3mm}
  \label{fig:factorizedmlexp}
\end{figure}

\begin{figure}
  \centering
      \includegraphics [width=0.45\textwidth]  {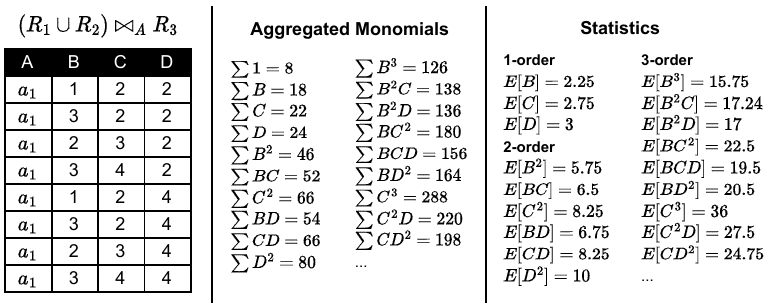}
      \vspace*{-3mm}
  \caption{\revise{Aggregated monomials and statistics.}}
  \vspace*{-4mm}
  \label{fig:materialized}
\end{figure}

\begin{example}
\label{example:fac}
\revise{Consider  $R_1, R_2, R_3$ in \Cref{fig:factorizedmlexp}. We aim to train linear regression on $(R_1{\cup} R_2){\Join_A} R_3$ using D as the feature and C as the target variable.
The naive solution is to first materialize the union and join results (\Cref{fig:materialized}) and then compute $\mathbf{X}'^T\mathbf{X}'$. 
Using factorized linear regression, we can optimize the query plan (\Cref{fig:factorizedmlexp}) by pushing down aggregations: $\gamma\left( (\gamma_A(R_1){\cup}\gamma_A(R_2)){\Join_A\gamma_A(R_3)} \right)$. This approach yields the same result as the naive solution, but avoids the costly materialization.}
 We use the aggregates to fit the linear regression:
{\small\begin{equation*}
    \theta = (\mathbf{X}^T\mathbf{X})^{-1}\mathbf{X}^T\textbf{y} =
    \begin{bmatrix}
    \sum D^2 & \sum D \\
    \sum D & \sum 1 
    \end{bmatrix}^{-1} \begin{bmatrix}
    \sum CD  \\
    \sum C
    \end{bmatrix} = \begin{bmatrix}
    80 & 24 \\
    24 & 8 
    \end{bmatrix}^{-1} \begin{bmatrix}
    66  \\
    22
    \end{bmatrix}
\end{equation*}}
\end{example}

After obtaining the model parameters $\theta$, the model performance can also be evaluated. For square loss, $\sum (y - x^T\theta)^2 = \sum (y^2 - 2\theta^Tx y + \theta^T xx^T \theta) = \textbf{y}^T\textbf{y} - 2\theta^T \textbf{X}^T\textbf{y} + \theta^T \textbf{X}^T\textbf{X} \theta$. The final aggregation result provides the necessary statistics to compute this expression.

\subsection{Monomial Semi-ring}
\label{sec:suffstat}

This section introduces sufficient statistics as vectors of monoids and extends it with semi-ring operations $+$ and $\times$. This helps bridge ideas from two communities---semi-rings from the factorized ML literature that train models over joins and unions, but primarily focused on non-private linear regression, and privatized sufficient statistics from the ML literature~\cite{huggins2017pass,kulkarni2021differentially,wang2018revisiting} that approximate generalized linear models, but do not support joins and unions.  
 We are the first to explicitly extend semi-ring from gram matrix (linear regression) to higher order monomial (generalized linear models).
\revise{This section focuses on the semi-ring design of monomials to support join and union operations {\it without \texttt{DP}}. In the next section, we introduce \fdp, a mechanism to privatize these monomials for \texttt{DP}.}

We first define the $k$-order monomial~\cite{huggins2017pass} in sufficient statistics:
\begin{definition}[$k$-order Monomial]
\label{def:monomial}
Given n random variables $f_1, f_2, ..., f_n$, the $k$-order monomials are random variables of monomials of the form $p =f_1^{k_1} f_2^{k_2}... f_n^{k_n}$, 
where $k_1, k_2 ..., k_n$ are $n$ non-negative integers such that $\sum_{i=1}^n k_i = k$.
\end{definition}

The core statistics to compute for ML are the expected value of each monomial $E[p]$. For example, $1$-order monomials estimate means, $1,2$-order monomials estimate covariance (core sufficient statistics for linear regression), and $1,2,3$-order monomials estimate skewness. Moreover, a generalized linear model can be approximated by high-order monomials using Taylor series expansions~\cite{huggins2017pass}.

\begin{example}
\label{exp:monomial}
    \revise{Consider the relation in \Cref{fig:materialized} (left) and random variables $B,C,D$. The  1-order monomials are $B, C, D$, the 2-order monomials are $B^2, BC, C^2, BD, CD, D^2$, and the 3-order monomials are $B^3,B^2C,B^2D,BC^2,BCD,...$. The statistics (right) are
    the expected monomials when the relation is the population, and can be derived from the aggregated monomials (middle). The 1,2-order statistics are the sufficient statistics for linear regression training (\Cref{example:fac}).}
\end{example}

\revise{Instead of computing statistics over join and union through costly materialization and subsequent aggregation, factorized linear regression utilizes semi-ring operators for $+$ and $\times$ to push down the aggregation of 1,2-order statistics. We extend this concept by defining operators for a $k$-order monomial semi-ring, thus generalizing factorized linear regression (2-order monomial semi-ring).}

\begin{definition}[$k$-order monomial semi-ring]
Given $m$ features $f_1, f_2, ..., f_m$, the $k$-order monomial semi-ring has domain of a vector with size $\frac{1-m^k}{1-m}$ for $m{\ge}2$, and $1{+}k$ for $m{=}1$. The domain  breaks into $k{+}1$ subvectors $[s_0, s_1, ..., s_k]$ where $s_i$ is a vector of size $m^i$. Then, given two semi-ring element $a = [s^a_0, s^a_1, ..., s^a_k]$ and $b = [s^b_0, s^b_1, ..., s^b_k]$, let:
\begin{align*}
a + b =&  [s^a_0 + s^b_0, s^a_1 + s^b_1, ..., s^a_k+s^b_k] \\
a \times b =& [s^a_0 \otimes s^b_0, \sum_{i=0}^1 s^a_i \otimes s^b_{1-i}, ..., \sum_{i=0}^k s^a_i \otimes s^b_{k-i}]
\end{align*}
where $\otimes: R^p \times R^q \rightarrow R^{pq}$ is the tensor product defined as:
for $\mathbf{a} = [a_1, a_2, ..., a_p]$ and $\mathbf{b} =[b_1, b_2, ..., b_q]$,  tensor product computes the pairwise product $\mathbf{a} \otimes \mathbf{b} = [a_1b_1, a_1b_2, ..., a_pb_1, a_pb_2, ..., a_pb_q]$.

The zero element is a vector of all zeroes, and the one element is a vector with non-zero $s_0 = [1]$, but the rest as all zeroes.
\end{definition}

\begin{figure}
  \centering
      \includegraphics [width=0.4\textwidth] {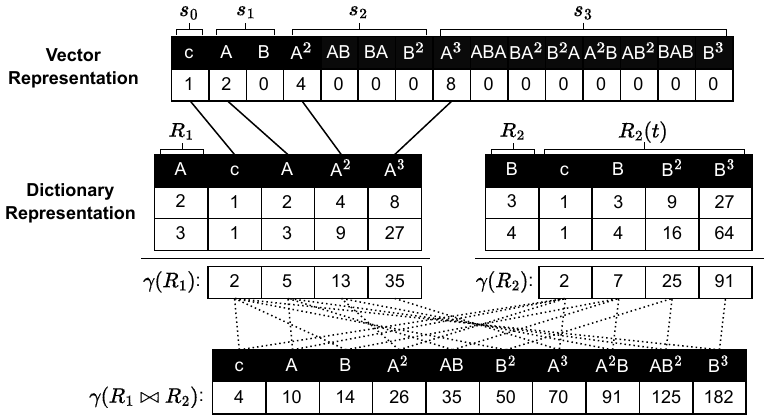}
      \vspace*{-2mm}
  \caption{Aggregation of $3$-order monomial semi-ring over join: $\gamma(R_1[A] \Join R_2[B])$.  Each row is one tuple, and we show the vector representation for the first tuple in $R_1$.  Dictionary representation removes redundancy and sparsity. Dotted lines map the contributing components to aggregated results.}
  \vspace*{-6mm}
  \label{join_example}
\end{figure}

Intuitively, each subvector $s^a_k$ holds the $k$-order monomials with a size $m^k$, as there are $m^k$ possible permutations with repetition.
In order to compute statistics using a $k$-order monomial semi-ring, we annotate $R$ by assigning to each tuple $t$ its monomials (non-existing features are considered to be all zeros).  
Note that, while this vector representation provides a straightforward way to define semi-rings for arbitrary orders, it is inherently inefficient and can be optimized by the dictionary representation discussed next.

\stitle{Dictionary Representation.} Vector representation has redundancies (e.g., $f_1f_2 = f_2f_1$) and sparsity (nonexistent features are zeros). Dictionary representations~\cite{khamis2018ac} help reduce redundancy: monomials serve as keys to deduplicate, and monomials with zeros are not materialized. 
We next provide an example of semi-ring operations  using the dictionary representation for join-aggregation:

\begin{example} Consider two relations of a single feature $R_1[A] = [2,3]$ and $R_2[B] = [3,4]$, and the aggregation query $\gamma(R_1\Join R_2)$ for $3$-order monomial semi-ring. \Cref{join_example} illustrates the annotated relations and the query processing. To start, the aggregations are pushed down by summing each monomial. Next, the monomials are combined according to the multiplication operator for join.
\end{example}

Assuming the join result, $R_\Join$, as the population, we can use the aggregated monomials to compute statistics (i.e., the expected monomials). Let $s{=}\gamma(R_\Join)$ be the aggregated monomial semi-ring. Then, for monomial $p$: $E[p] = s[p]/s[c]$, where $4$ is the count ($0$-order monomial). For example, in \Cref{join_example}, $E[AB]{=}s[AB]/s[c] {=} 35/4\\ {=} 8.75$.
\revise{The aggregated monomials comprise count and sum aggregations over the base tables, which can be efficiently computed by requesters/providers using SQL  queries. Further, they serve as an ideal intermediate for \texttt{DP} due to their reusability, as discussed next.}

\subsection{\fdp Mechanism}
\label{sec:fdp}

\revise{In this section, we present the Factorized Privacy Mechanism (\fdp) which applies the  Gaussian mechanism to the aggregated monomials discussed in the previous section to support differentially private data search (\Cref{sec:problem}) while maintaining high utility. The primary algorithmic challenge \fdp addresses is designing sufficient statistics that are composable (through semi-ring operators) and reusable (as post-processing without additional privacy cost) to support ML across various join and union augmentations.
}

We make the following simplifications: (1). Features consist only of numerical attributes, and join keys consist only of categorical attributes. \revise{\Cref{sec:system} describes preprocessing to support categorical features.}
(2).
Group-by operator $\gamma$ has been extended to annotate group-by keys without tuples with zero elements (group-by attribute domains are assumed public in \Cref{sec:problem}).
(3). Datasets are preprocessed so that the $\ell_2$ norm of the features in each tuple is bounded by a constant value $B$, following previous works~\cite{wang2018revisiting,dwork2014analyze}.

\stitle{Algorithms.} The \fdp mechanism, detailed in \Cref{alg:fdp}, is applied locally by either the requester or provider before dataset upload to the search platform. It uses as inputs: (1) the relation $R$ to be privatized  (2) the join key $A$\footnote{$A$ could be composite. To support multiple join keys, the \texttt{DP} budget can be split among different key combinations. Additionally, optimization techniques can be applied to take advantage of the correlations between join keys~\cite{qardaji2014priview}.}, 
which $= null$ if $R$ is only for union, (3) the order of monomials, $k$, based on the model to support, and (4) the \texttt{DP} budget $(\epsilon,\delta)$ for $R$. 
\revise{\fdp computes locally aggregated monomials $\gamma(R)$ and applies the Gaussian mechanism to these, with sensitivity optimized based on order parity and feature count (line 2, 7):  For even-order monomials, sensitivity is reduced by $\sqrt{2}$, and if there's only one feature, sensitivity is reduced by another $\sqrt{2}$.
}

\begin{theor}
\label{theory:fdp}
\fdp is $(\epsilon,\delta)-DP$.
\end{theor}
\begin{hproof} 
\fdp applies the Gaussian mechanism~\cite{dwork2006our} to the aggregated monomials for $(\epsilon,\delta)-DP$. Therefore, we only need to show the correctness of $\Delta$. We present simple cases illustrating proof concepts for the union and join of 1 feature (with lower $\Delta$), and the union of 2 features with $1/2$-order monomial semiring. These cases are meant to illustrate the key intuitions; full proofs and generalizations are available in \Cref{sensitivity} due to space limits. 

\noindent $\bullet$ {\it (1 feature, Union, any order)} For union, count ($0$-order monomial) remains unchanged as we consider bounded \texttt{DP}, where the neighbour relation has one tuple modified (instead of removed/added).
Let the modified feature value be $a \rightarrow a'$ where both $a$ and $a'$ have a domain of $[-B,B]$.  Then, for the $i$-th monomial, the squared difference is $(a^i - a'^i)^2$. When $i$ is odd, $a^i \in [-B^i,B^i]$, and $(a^i - a'^i)^2\leq (2B^i)^2$. When $i$ is even, $a^i \in [0,B^i]$, and $(a^i - a'^i)^2\leq (B^i)^2$.

\noindent$\bullet${\it (1 feature, Join, any order)} For join, the query also groups results by join key $A$. This can be considered as a histogram~\cite{xu2013differentially}, where each bin is a join key, and the value is the $k$-order monomial semi-ring. The neighbouring relation has two cases: the modified tuple has changed the join key or not. If the join key doesn't change, this is the same as the union case. If the join key changes, there are two bins with a maximum square difference of $\sum_{i=0}^{k} B^{2i}$ (note that, unlike the union, the counts change). Thus, the sensitivity is bounded by $\sqrt{2\sum_{i=0}^{k} B^{2i}}$. Finally, we take the maximum.

\noindent$\bullet${\it (2 features, Union, $1$-order)} Let the modified feature value be $(a,b) \rightarrow (a',b')$ where both $a^2 + b^2$ and $a'^2+b'^2$ are $\leq B^2$. Then, consider the $1$-order monomials $(a,b),(a',b')$. The squared difference is:
\begin{align*}
(a - a')^2 + (b-b')^2 
\leq & (2a^2 + 2a'^2) + (2b^2 + 2b'^2) = 4B^2 
\end{align*}

The sensitivities for higher odd orders are similar.

\noindent$\bullet$ {\it (2 features, Union, $2$-order)} For even-orders, we can obtain a tighter bound. Consider the $2$-order monomials $(a^2,ab,b^2),(a'^2,a'b',b'^2)$. The squared difference is:
\begin{align*}
     & (a^2 - a'^2)^2 +(ab-a'b')^2 +(b^2-b'^2)^2 \\
\leq & (a^2 - a'^2)^2 +2(ab-a'b')^2 +(b^2-b'^2)^2 \\
= & (\blue{a^4}{-}\red{2a^2a'^2}{+}a'^4){+}(\blue{2a^2b^2}{-}\red{4aba'b'}{+}2a'^2b'^2){+}(\blue{b^4}{-}\red{2b^2b'^2}{+}b'^4)\\
= & \blue{(a^2 +b^2)^2} + (a'^2 +b'^2)^2 - \red{2(aa'+bb')^2} \leq \blue{B^4} + B^4 - \red{0} = 2B^4
\end{align*}

For higher even orders, we can similarly amplify the monomials by the binomial coefficients (second line) to find a non-negative red term for even-order monomials, resulting in a tighter bound. Extending to joins follows a similar approach as the single feature case, where we consider group-by queries as histograms.
\end{hproof}

\input{sections/fdpalg.tex}

\subsection{Comparison with Other Mechanisms}
\label{sec:utility}

We next analyze the error of \fdp in estimating the statistics $s$ (expected values of monomials). Generally, the expected errors of $s$ are correlated with the error of
the target model parameter $\beta$ and accuracy; we will study the confidence bound for linear regression parameter in the next section, where the $s$ error is the key factor.

\stitle{Setting.} We consider a data corpus with size $n_{corp}$ (defined as the number of provider datasets) and has received $n_{req}$ requests. To simplify the analysis, we assume that: (1) the search only uses union operations (and we will discuss the extension to join). (2) each dataset has one feature, $n$ tuples, and a \texttt{DP} budget of $(\epsilon,\delta)$. The search platform evaluates all possible augmentations, each corresponding to a unique combination of provider datasets.

\stitle{Metrics.} 
The goal is to evaluate, for each augmentation, the expected $\ell_2$ error of the privatized set of monomials $\tilde{s}$:  $E[\|s - \tilde{s}\|_2]$.

\stitle{Mechanisms.} \revise{We compare \fdp  with standard \texttt{DP} mechanisms used in various existing trust models:}

\begin{myitemize}
\itemsep0em 

  \item For \sys's trust model (\Cref{sec:problem}), \fdp (\Cref{alg:fdp}) privatizes local aggregates independently for each dataset, and combines the aggregates with factorized ML.

  \item \revise{For the local model, the Per-tuple Privacy Mechanism (\texttt{TPM})\footnote{An alternative is to apply Gaussian mechanism to raw tuples and then compute monomial semi-ring; this, however, results in an even larger error.} applies \Cref{alg:fdp} to privatize each tuple~\cite{yang2020local}.}

  \item \revise{For the global model, the Aggregate Privacy Mechanism (\texttt{APM})\footnote{There are other alternatives like perturbing objectives and gradients; however they are similarly limited by the combinatorially large number of models to train.} first computes the union result $R_{\cup}$ after augmentation, and then applies \Cref{alg:fdp} to  $\gamma(R_{\cup})$~\cite{wang2018revisiting}. To ensure $(\epsilon,\delta)-DP$ for all $n_{req}(2^{n_{corp}-1}-1)$ augmentations, the \texttt{DP} budget has to be split.}

  \item \revise{For the shuffle model, shuffling~\cite{erlingsson2019amplification} privatizes each tuple, similar to \texttt{TPM}, but applies Laplace mechanism with the amplified privacy budget. These tuples are shuffled either at the $1^{st}$- (\texttt{SF-1}) or $2^{nd}$-level (\texttt{SF-2}); akin to \texttt{APM}, \texttt{SF-2} requires budget splits.}
\end{myitemize}

\begin{prop}
\label{theor:utility}
For the estimation of each augmentation (assuming that the number of augmented datasets and the order of $s$ are small constants), \fdp/\texttt{SF-1} has expected $\ell_2$ error of $\Tilde{O}(\Delta/n\epsilon)$, while \texttt{TPM} has an error of $\Tilde{O}(\Delta/\sqrt{n}\epsilon)$ and \texttt{APM}/\texttt{SF-2} has an error of $\Tilde{O}(n_{req}2^{n_{corp}}\Delta/n\epsilon)$, where $\Tilde{O}(\cdot)$ hides at most a logarithmic term.
\end{prop}
The proof is in \Cref{utility}.

\stitle{Remark.} \Cref{theor:utility} highlights prior mechanisms' limitations: \revise{\texttt{APM}/\texttt{SF-2} are competitive only for small corpora and quickly exhaust budget for larger requests/corpus sizes due to budget split for all possible augmentations, and require trust in centralized aggregators/shufflers.}  \texttt{TPM} adds excessive noise to each tuple, requiring quadratically more tuples to achieve the same level of error as \fdp. Although \texttt{SF-1} can theoretically match \fdp's complexity with privacy amplification,  it's significant only for large numbers of tuples. For instance, given $\epsilon{=}1$ and $\delta{=}10^{-6}$, $\epsilon$ is amplified when $n$ reaches ${\sim}650$~\cite{erlingsson2019amplification,feldman2022hiding}. However, small $n$ needs amplification most, where \texttt{SF-1} provides much larger errors than \fdp (\Cref{sec:abalationexp}).

Extending the analysis to joins involves considering group-by errors based on domain size and multiplication of privatized monomials. Comparisons remain similar: \texttt{TPM} needs a quadratically larger data size, while \texttt{APM} may outperform \fdp only for small corpora and requests but exhausts budget for larger corpus sizes.

\vspace{-2mm}
\subsection{Differentially Private Data Search Platform}
\label{sec:system}

In this section, we discuss \sys, a data search platform that integrates \fdp to ensure differential privacy.

\stitle{Provider.}  The architecture of the \sys, which uses \fdp for \texttt{DP}, is illustrated in \Cref{fig:arch}. For each dataset $R$ data provider owns, the supported operation ($\Join/\cup$\footnote{Any dataset supports join also supports union by aggregating out the join key.} or $\cup$-only) is decided.
If join is supported, the join key $A$ must also be specified. 
\fdp is then applied locally to $R$ to privatize the sufficient statistics $\widetilde {\gamma(R)}$, which are then uploaded to \sys. As \sys is not trusted, {\it data storage} only stores privatized statistics, but not raw data. All operations over $\widetilde{\gamma(R)}$ are post-processing without additional \texttt{DP} costs. 

\stitle{Requester.} The requester has model type $M$ and $R_{train}$, and wants to improve accuracy on $R_{test}$. The requester computes and submits to \sys the privatized sufficient statistics $\widetilde{\gamma(R_{train})}$ and $\widetilde{\gamma(R_{test})}$.  {\it Data discovery} returns a set of joinable or unionable relations $R$ from {\it data storage}.  Then, {\it Data search}  applies greedy algorithm (following Kitana~\cite{kitana}): in each iteration, it evaluates each candidate and adds the one that most improves the model accuracy.  \sys is agnostic to the search algorithm, and others~\cite{vafaie1994feature,chepurko2020arda} can also be used.

\stitle{Data Discovery.} Previous data discovery systems~\cite{castelo2021auctus,fernandez2018aurum} leverage MinHash sketches, column type and data distribution statistics; \sys supports all of them. Specifically, for categorical attributes, we utilize minhash sketches, computed from public domains, to measure set similarity. For numerical attributes, we rely on public schemas for column names and types. Additionally, we construct (approximated) data distribution statistics such as count, mean, standard deviation, and correlation from the privatized $2$-order monomial semi-rings, without additional \texttt{DP} costs.

\stitle{Preprocessing.} Before applying \fdp, requesters and providers can locally preprocess datasets to enhance utility and robustness.  \revise{For instance, datasets may have categorical features not directly supported by the proxy model (linear regression).  For low cardinality categories, standard one-hot encoding can be applied, treating the encoded features as numerical for privatization by \fdp. However, high cardinality categorical features yield high-dimensional vectors when one-hot encoded, which is problematic and typically requires specialized techniques~\cite{moeyersoms2015including,cerda2020encoding}.
This is precisely the problem \sys can address through augmentation. By joining with augmentations, high cardinality categories in $R_{train}$, like location, can be encoded into meaningful lower dimensional numerical features, like population and economic indicator, from augmented relations. Hence, we suggest using high cardinality categorical features as join keys.}

\sys also applies two steps to boost \texttt{DP} robustness.  First, it removes outliers (${>}1.5$ std from the mean), which typically improves model performance and reduces the tuple $\ell_2$ norms, enhancing \texttt{DP} noise robustness~\cite{lee2011much}.   
Second, {\it all \texttt{DP} mechanisms} (including ours) degrade with increasing dimensionality due to the increased tuple $\ell_2$ norms.
Thus, \sys applies dimensionality reduction~\cite{mackiewicz1993principal} to retain the top $K$ principal components ($K{=}1$ works best in our experiments), and rescales tuples to bound max $\ell_2$ norm  ${\leq}B$. This lowers the noise scale, improves utility, and achieves a lower sensitivity $\Delta$ with \#fea$=1$ (\Cref{alg:fdp}).
These steps are applied to all datasets and \texttt{DP} baselines in our real-world experiments (\Cref{exp:nyc}).

\stitle{Supporting Varied Privacy Needs.}
A unique benefit of \sys's design is that it can  adapt to different privacy needs.   In cases where pure \texttt{DP} ($\delta{=}0$) is required, \fdp can be modified to apply Laplace mechanisms~\cite{dwork2006calibrating}.
In situations where individuals don't trust providers or requesters, \fdp can be reduced to  \texttt{LPM} to privatize individual tuples. Conversely, shuffling only guarantees approximate \texttt{DP}  and \texttt{GPM} always requires a trusted centralized aggregator.

\stitle{\revise{ML training after data search.}} 
\revise{
After \sys finds predictive augmentations using a differentially private proxy model (linear regression), the model could be directly returned to requesters. However, requesters may need more complex model $M$, and the training shall also satisfy \texttt{DP}.
To achieve this, \sys can be integrated within a larger differentially private federated ML system ~\cite{wei2020federated,truex2020ldp,zhao2020local,wang2020hybrid}, where \sys first locates augmentations, and then the ML systems use the augmented dataset to train sophisticated models, such as deep neural networks, via differentially private gradient descent.
}

\stitle{Scope.} While \sys can employ \fdp to support a wide range of models~\cite{schleich2019layered} and approximate GLM~\cite{kulkarni2021differentially}, this paper focuses on linear regression~\cite{schleich2016learning} because it's widely used and is adopted by previous data search~\cite{kitana,chen2017semi,santos2022sketch}. Next, we dive deep into linear regression to analyze the task utility and propose further optimizations.

%% file: sections/fdpalg.tex
\begin{algorithm}

\caption{\fdp mechanism}
\label{alg:fdp}
\SetKwInOut{Input}{inputs}
\SetKwInOut{Output}{output}

\Input{Relation $R$, Join Key $A$, Order $k$, \texttt{DP} budget $(\epsilon,\delta)$}
\Output{Privatized Aggregated Relation $\tilde{R}$}
\eIf{$A = null$ (Union Only) }{ 
  $\Delta = \sqrt{\sum_{i=1}^{k} (\text{if i odd: 4, elif \#fea=1: 1, else: 2}) \cdot B^{2i}}$\;
  $\sigma, \tilde{R} = \sqrt{2\ln(1.25/\delta)}\Delta/\epsilon, \gamma(R)$\;
  // add i.i.d. noises to each $1-k$ order monomial $s$\;
  $\tilde{R} = \{s: \tilde{R}[s] + e {\sim}\mathcal{N}(0, \sigma^2)$ for $1-k$ monomial $s$\}\;
}    
{
  {\footnotesize $\Delta{=} max(\sqrt{\sum_{i=1}^{k} (\text{if i odd: 4, elif \#fea=1: 1, else: 2}) \cdot B^{2i}}, \sqrt{2\sum_{i=0}^{k} B^{2i}})$\;}
  $\sigma, \tilde{R} = \sqrt{2\ln(1.25/\delta)}\Delta/\epsilon, \gamma_A(R)$\;
  \ForEach{$a \in dom(A)$}{
  // add i.i.d. noises to each $0-k$ order monomial $s$\;
  $\tilde{R}(a) {=} \{s{:}\tilde R(a)[s]  {+} e {\sim}\mathcal{N}(0, \sigma^2)$ for $0-k$ monomial $s$\}\;
  }
}
\Return $\tilde{R}$\;
\end{algorithm}
\vspace{-5mm}

%% file: sections/analysis.tex
\vspace{-2mm}
\section{Diving Deep Into Linear Regression}

This section examines the ML task {\it utility} \fdp provides and suggests optimizations for linear regression. 
We start with the assumption of linear regression on many-to-many join (as opposed to one-to-one~\cite{wang2020hybrid,hardy2017private}), which is challenging due to unexpected duplication and independence. We then propose an unbiased estimator. Next, we explore the confidence bounds for the linear regression parameters and propose optimizations to tighten the bound further.

\vspace{-2mm}
\subsection{Linear Regression on Many-to-Many Join}
\label{sec:joinassumption}

Linear regression assumes a noisy linear relationship between the features and target variable: $\textbf{y} = \textbf{X}\beta + \textbf{e}$, where $\textbf{e}$ is the error term. This is consistent with our assumption so far if $R_\Join = R_1 \Join \dots R_k$ is the population,  and let us use the monomial semi-ring to compute the expected $s$. 
However, when many-to-many joins are involved, $R_\Join$ often doesn't represent the population as joins generate Cartesian products for each matching key.   
This leads to (1) duplicated tuples (the same $\textbf{y}$ values are repeated) and (2) unexpected {\it independence} between features from different relations with the same join key, leading to biased estimation.

To the best of our knowledge, linear regression over many-to-many joins has been understudied. The closest work is multi-view learning~\cite{frank2007method,guo2008multirelational}, which pre-aggregates (e.g., averaging) features. However, this introduces errors for long join paths due to Simpson paradox~\cite{pearl2022comment} (e.g., average of average is not average). In contrast, we propose an unbiased estimator based on the assumptions from vertical federated ML that each party holds a projection; this complements prior factorized ML work~\cite{abo2016faq,olteanu2015size,weighing}, which studied the computational complexity of many-to-many joins.

Our analysis focuses on an easy-to-explain case inspired by vertical federated ML~\cite{wang2020hybrid,hardy2017private}, where we want to train linear regression over relation $R$. However, $R$ is not directly observable, and each party can only access a projection $\pi(R)$. Multiple $\pi(R)$ may have many-to-many relationships on the common attribute (join key) instead of the one-to-one relationships studied by federated ML. The objective is to train linear regression on $R$ collectively.

\stitle{Unbiased Estimator.} Given $R$ of cardinality $n$, suppose there are two parties holding different projections $\pi_{F_1}(R)$ and $\pi_{F_2}(R)$, and the goal is to compute the 2-order monomial semi-ring $\gamma(\pi_{F_1\cup F_2}(R))$. However, factorized ML is trained on $R_\Join = \pi_{F_1,J}(R) \Join_J \pi_{F_2,J}(R)$ with join key $J = F_1\cap F_2$; $R_\Join$ is likely to differ from $\pi_{F_1\cup F_2}(R)$ (unless $J$ is primary key), resulting in bias. To address this, we propose an unbiased estimator for $s$ based on $s' = \gamma(R_\Join)$.
\begin{prop}[Unbiased Estimator of $s$ over $R$]
We make the simplifying assumption that $J$ is uniformly distributed (if $d=|dom(J)|$, each $j \in J$ appears $n/d$ times in $R$) and is not correlated with any other attribute.  Let $s' = \gamma(R_\Join)$. Then,
\begin{equation*}
\hat{s} = \begin{cases}
\mathrlap{f_1f_2 = \frac{1-n}{1-d} \frac{s'[f_1f_2]}{s'[c]} + \frac{n-d}{1-d} \frac{s'[f_1]}{s'[c]} \frac{s'[f_2]}{s'[c]}}\\
&\text{for $f_1 \in F_1 - J, f_2 \in F_2 - J$}\\
p = s'[p]/s'[c] &\text{for any other monomial $p$}\\
\end{cases}
\end{equation*}
$\hat{s}$ is an unbiased estimator of  monomial semi-ring $s= \gamma(\pi_{F_1\cup F_2}(R))$.
\end{prop}
\vspace{-2mm}
The proof is in \Cref{sec:unbiasedproof}. We assume vertical partitions of $R$, but real-world datasets may also be horizontally partitioned; the estimators could be refined for these cases. Our analysis studies the base case, and the unbiased estimator can be recursively applied for multiple joins and unions. Note that the estimators are post-processing steps without compromising \texttt{DP}.

\subsection{Simple Linear Regression Analysis}
\label{sec:betaboundanalysis}
Building on the assumption in the previous section, this section studies the confidence bound of factorized linear regression. Compared to~\cite{wang2018revisiting}, our analysis focuses on simple linear regression with one feature, under less stringent assumptions; this scenario is sufficient to show  \fdp's advantages over other mechanisms, and motivates optimization. We first consider a single relation case, then extend to union and join. We'll begin with defining the confidence bound, which will be used to evaluate the utility of private estimators.

\begin{definition}[Confidence Bound]
Given parameter $\theta$, the $(1-p)$ confidence bound $C^{\tilde{\theta}}_{\hat{\theta}}(p)$ for an private estimator $\tilde{\theta}$ is:
\vspace{-2mm}
  $$C^{\tilde{\theta}}_{\hat{\theta}}(p) = \inf\,\,\{b: \pr[|\tilde{\theta} - \hat{\theta}|\leq b] \geq 1 - p\}$$
\vspace{-2mm}
where $\hat{\theta}$ is the non-private estimator.
\end{definition}

We consider relation $R[x,y]$ with one feature $x$, target variable $y$, and cardinality $n$. We want to train $y = \beta_x\cdot x + \beta_0$,  and focus on the parameter $\beta_x$; $\beta_x$ has an optimal non-privitized estimator  $\hat{\beta_x} = \frac{\widehat{E[xy]} - \widehat{E[x]}\widehat{E[y]}}{\widehat{E[x^2]} - \widehat{E[x]}^2} = \cov/\var$, where $\cov$ and $\var$ are polynomials that can be derived from aggregeted 2-order monomials $\gamma(R)$. 
We apply \fdp to compute the privatized 2-order $\gamma(R)$ and study the confidence bound of the privatized estimator $\tilde{\beta}_x$. 
Note that more familiar error definitions like mean-squared-error can be upper bounded, roughly, by the square of the confidence bound.

\begin{theorem}[Confidence Bound of $\tilde{\beta}_x$]\label{theo:lrcbound}
For every $p$ where $\tau_1< 1$ holds, the $(1-p)$ confidence bound for $\tilde{\beta_x}$ is:
\vspace{-2mm}
$$
C^{\tilde{\beta_x}}_{\hat{\beta_x}}(p) \leq \tau_2 + \frac{\tau_1}{1-\tau_1}\left(\hat\beta_x + \tau_2\right)
$$
where $\tilde{\beta}_x$ ($\hat\beta_x$) is the private (non-private) estimate of  $\beta_x$. Let $B_1$ and $B_2$ be the $(1-p)$ confidence bounds  for $\pvar$ and $\pcov$ respectively. Then $\tau_1 = B_1/\var$ and $\tau_2 = B_2/\var$ are both $O\left(\frac{B^4\ln(1/\delta)\ln(1/p)}{\epsilon^2 n\var}\right)$. The probability is taken over the randomness of \fdp.

\label{thm:pfactorized}
\end{theorem}

The proof and extension to multi-features can be found in \Cref{app:linear}. \Cref{theo:lrcbound} demonstrates that the private estimator $\tilde{\beta_x}$ is asymptotically close to the non-private $\hat{\beta_x}$. The key factors in reducing the discrepancy are $\tau_1,\tau_2$. \texttt{APM} and \texttt{SF-2} have combinatorially large $\tau_1,\tau_2$ due to the  budget splits. \texttt{TPM} requires quadratically more data than \fdp to achieve the same level of $\tau_1,\tau_2$.

\stitle{Extension to Factorized ML.}
The full procedures to extend the confidence bounds for factorized ML are in \Cref{app:linear}. For the union of $k$ datasets: $R = R_1 \cup R_2... \cup R_k$, $\tau_1$ and $\tau_2$ are reduced by a factor of $\sqrt{k}$, while the rest remain unchanged. For the join of two datasets $R[x,y,J] = R_1[x,J] \Join R_2[y,J]$, where $|dom(J)|=d$, there is additional noise to the count which could cause distortion if the privatized count is close to or less than $0$. To address this, an additional assumption that noises to count is $o(n/d)$ is needed~\cite{wang2018revisiting}, resulting $\tau_1$ and $\tau_2$ to increase by a factor of $O(\sqrt{d})$ and $O(n/\sqrt{d})$.

\input{sections/fdpsmart.tex}

\subsection{Optimization: Better Noise Allocation}
\label{opt:noiseallocation}

In \Cref{sec:betaboundanalysis}, we analyzed the linear regression confidence bounds.  We propose to adjust noise allocation to improve the bounds further.

First, previous work (e.g.,~\cite{Alabi22}) has shown that $\beta_x$ is usually the parameter of interest instead of $\beta_0$ for linear regression over the union.  In this case, we suggest each provider adding noises directly to $\sigma_x^2, \sigma_{xy}^2$, rather than monomials $x^2, xy, x, y$. This reduces $\tau_1$ and $\tau_2$ by a factor of $O(B^2\sqrt{\ln(1/\delta)\ln(1/p)}/\epsilon)$ (\Cref{app:noiseopt}).

Second, optimizing joins is more difficult as we add noise locally to monomials to circumvent combinatorially large \texttt{DP} costs. However, we can reduce $\tau_1,\tau_2$ by $O(B^2)$ through smart budget allocation (\Cref{app:noiseopt}). Our insight is that
lower-order monomials are multiplied by more monomials than higher-order ones. 
For example, in \Cref{join_example}, $0$-order monomials are multiplied by $0,1,2,3$-order ones, while $3$-order monomials only multiply with $0$-order ones. Hence, we shall decrease the noise to lower-order ones. \fdpopt in \Cref{alg:fdpopt} achieves this by dividing the \texttt{DP} budget across orders; lower order monomials have lower sensitivity and thus fewer noises.

%% file: sections/fdpsmart.tex
\vspace{-3mm}
\begin{algorithm}
\caption{\fdpopt algorithm for Join}
\label{alg:fdpopt}
\SetKwInOut{Input}{inputs}
\SetKwInOut{Output}{output}
\Input{Relation $R$, Join Key $A$, Order $k$, \texttt{DP} budget $(\epsilon,\delta)$}
\Output{Privatized Annotated Relation $\tilde R$}
\ForEach{$i \in \{0,\dots,k\}$}{
$\Delta, \epsilon', \delta' = (\text{if i odd: 2,  else:}\sqrt{2}) \cdot B^{i}, \epsilon/(k+1), \delta/(k+1)$\;
  $\sigma, \tilde R = \sqrt{2\ln(1.25/\delta')}\Delta/\epsilon', \gamma_A(R)$\;
  \ForEach{$a \in dom(A)$}{
  // add i.i.d. noises to each $i$ order monomial $s$\;
  $\tilde R(a) {=} \{s: \tilde R (a)[s] + e {\sim}\mathcal{N}(0, \sigma^2)$ for $i$ monomial $s$\}\;
  }
  }
\Return $\tilde R$\;
\end{algorithm}
\vspace{-6mm}

%% file: sections/experiments.tex
\vspace{-3mm}
\section{Evaluation}

We evaluate \fdp on  NYC Open Data~\cite{nycopen} corpus of 329  datasets for an end-to-end dataset search. We then use ablation studies via synthetic datasets to validate our theoretical analyses.

\input{sections/nyc_exp_fig.tex}
\input{sections/union_exp_fig.tex}

\vspace{-3mm}
\subsection{Real-world Experiments}
\label{exp:nyc}

\stitle{Data and Workload.} We construct a large data corpus of 329 NYC Open Data~\cite{nycopen} datasets.  Since prior DP mechanisms need to know the number of requests up front, we create a workload of 5 requests using the following random datasets:
\vspace{-1mm}
\begin{myitemize}
\item \textbf{Regents}~\cite{Regent} contains 2014-17 regents exams data.
\item \textbf{ELA}~\cite{ELA} contains 2013-18 Early Learning Assessment (ELA) data.
\item \textbf{Gender}~\cite{gender} contains 2013-16 ELA data by grades and gender.
\item \textbf{Grad}~\cite{Grad} contains 2016-17 graduation outcomes.
\item \textbf{Math}~\cite{Math} contains 2013-18 Math grades.
\end{myitemize}
\vspace{-1mm}
\noindent For each request, we look for a single dataset to join/union with the requested dataset.  We turn off data discovery so every dataset in the platform is considered.   
By default, each dataset has \texttt{DP} budget $(\epsilon=1,\delta=10^{-6})$. We report the final $r2$ score evaluated non-privately. For reliability, we run each request 10 times.

\stitle{Baselines.} ~\revise{We consider different \texttt{DP} mechanisms.
\textbf{\texttt{Non-P}} doesn't use DP and provides $r2$ upper bound. 
\textbf{\fdp} applies \Cref{alg:fdpopt} to each dataset.
\textbf{APM} (Aggregate Privacy Mechanism), following Wang~\citep{wang2018revisiting}, applies \Cref{alg:fdpopt} to the augmented dataset to privatize the aggregated sufficient statistics (and requires a trusted search platform). We use attribute max-frequence from Flex~\cite{johnson2018towards} to bound join sensitivity.   Note that \texttt{APM} requires budget splits across all augmentations. 
\textbf{TPM} (Per-tuple Privacy Mechanism) applies \Cref{alg:fdpopt} to each tuple and uses half the $\epsilon$ to perturb the join key with generalized random response~\cite{kairouz2016discrete}.
\textbf{SF} is similar to \texttt{TPM}, but applies the Laplace mechanism to each tuple with an amplified budget then shuffles~\cite{feldman2022hiding,erlingsson2019amplification}. Since \texttt{SF} doesn't support joins (by $2^{nd}$-level aggregator), we only shuffle each dataset locally by $1^{st}$-level aggregators.
In each case, we use a failure mechanism that reports $r2=0$ if the privatized $\widetilde{\textbf{X}^T\textbf{X}}$ is not positive definite~\cite{Alabi22}.}

\stitle{Results.}
\Cref{exp:nycutility} shows the non-private $r2$ of 10 runs of private data search for the 5 requests. \fdp dominates the DP alternatives and is ${\sim}50{-}90\%$ of the non-DP case. \fdp's performance depends on dataset cardinality: the {\it Gender} dataset contains on average ${\sim}40$ tuples per join key (compared to ${>}100$ tuples per join key in other datasets) and is more vulnerable to noise.

We next vary the number of datasets by sampling
$n_{corp}{\in}\{10,50,\\100,300\}$ datasets and rerunning each baseline over the smaller corpus.  \Cref{exp:nycrepo} reports the median $r2$. For a small corpus ($n_{corp}=10$), \texttt{APM} outperforms \fdp because it imputes noise to the aggregated statistics across join key values and there are fewer budget splits, while \fdp has to add noise to the individual statistics for each join key. \texttt{TPM} and \texttt{SF} have low $r2$ due to high noise.

Finally, we vary the number of requests ($n_{req}{\in}\{1,10,50,100\}$) by sending the same request $n_{req}$ times, and report median $r2$.  \Cref{exp:nycrequest} shows that each baseline is almost invariant to $n_{req}$, and \fdp dominates. In theory, \texttt{APM} is worse for more requests but is already poor due to the large dataset corpus.

\vspace{-3mm}
\subsection{Synthetic Dataset Experiments}
\label{sec:abalationexp}

We next validate our theoretical analysis of linear regression using synthetic data, and conduct ablation tests to study the impact of various parameters (number of tuples $n$, \texttt{DP} budget $\epsilon,\delta$, corpus size $n_{corp}$, number of requests $n_{req}$ and join key domain size $d$).

\vspace{-2mm}
\subsubsection{Setup} 
We generate datasets by first creating a symmetric positive-definite matrix ({\footnotesize{\verb|make_spd_matrix|}} in $sklearn$) as the covariance $\textbf{X}'^T\textbf{X}'$. We then sample from a multivariate normal distribution with this covariance to create a relation. To ensure the $\ell_2$ norms of tuples $\leq B{=}5$, we resample for any tuples that exceed this limit.

By default, for union, we generate relations with $n=1000$ tuples and $3$ numerical attributes $[y, x_1, x_2]$. For join, we generate relations with $n=10000$ tuples and include a categorical join key $J$ uniformly distributed with a domain size of $d=100$ . We construct two vertical partitions with projections $[y, x_1,J]$ and $[x_2,J]$, respectively.
We start with $n_{corp}=2$ datasets, $n_{req}=1$ request.

We will report the $\ell_2$ distance to the non-private sufficient statistics ($s$ error) and regression parameter ($\beta$ error) as metrics. Each experiment will be repeated $100$ times, and we will present the medians (dots), as well as the $25th$ and $75th$ percentiles (error bars).

 \vspace{-2mm}
\subsubsection{DP for Union}
Baselines include \texttt{APM}, \texttt{TPM} (same as in \Cref{exp:nyc}), \texttt{SF-1}, which shuffles tuples locally, \texttt{SF-2}, which shuffles the unioned dataset, and \fdp using \Cref{alg:fdp} rather than \Cref{alg:fdpopt} (which is for join).

First, we vary $n{\in}\{10,100,500,1K, 10K\}$. \Cref{fig:serror} and \Cref{fig:betaerror} report $s$ and $\beta$ errors. Since there are $n_{corp}{=}2$ datasets, \texttt{APM} and \fdp perform similarly. In contrast, \texttt{TPM} requires quadratically more data to achieve the same $s$ errors, consistent with our analysis in \Cref{sec:utility}. \texttt{SF}'s amplification is not significant for small $n$, when it's needed most, and both variants have high $s$ errors.
$\beta$ error eventually converges to $0$ for all baselines, but \fdp does so at a comparable rate to \texttt{APM} ($n{=}500$ vs. $10K$ for the others).

\Cref{fig:fs_relation} shows that $\beta$ error naturally correlates with $s$ error, and higher $s$ error increases the chance of failure ($\beta$ error ${=}\infty$).  The remaining results will focus on $\beta$ error, as it is  of interest.

Next, we vary the \texttt{DP} budget $\epsilon{=}0.1$ or $\delta{=}0$ (pure \texttt{DP}). The results are shown in \Cref{fig:betaerror_small_eps} and \Cref{fig:betaerror_pure}, respectively. 
For $\epsilon{=}0.1$, the plot shifts right due to a smaller budget. In the case of pure \texttt{DP} with $\delta{=}0$, \fdp, \texttt{APM} and \texttt{TPM} can adapt to it by applying Laplace mechanism, achieving similar performance. In contrast, \texttt{SF-1} and \texttt{SF-2} fail as only approximate \texttt{DP} is supported.

\Cref{fig:betaerror_users} and \Cref{fig:betaerror_requests} vary the number of datasets and requests $n_{corp}, n_{req}{\in}\{1, 5, 10, 50, 100\}$, respectively. \fdp's $\beta$ error is flat. \texttt{TPM}, \texttt{SF-1} and \texttt{SF-2} frequently fail due to high noise, while \texttt{APM} only performed well when $n_{corp}{\leq}5$ or $n_{req}{\leq}10$. 
\textbf{\texttt{APM} is hence unsuitable for large data corpora.}

\Cref{fig:beta_error_unionopt} reports linear regression optimization benefit in \Cref{opt:noiseallocation}.  For a two-attribute dataset $R[y,x]$, while \fdp adds noise to monomials ($x,y,x^2,y^2,xy$), \fdpopt  adds noise to polynomials ($\sigma_{x}^2,\sigma_{xy}^2$) because we only care about $\beta_x$. We find that \fdpopt reduces the $\beta$ error and failure likelihood, especially for $n{<}100$.

\vspace{-2mm}
\subsubsection{DP for Join}
We evaluate different \texttt{DP} mechanisms over the join. 
Baselines include \fdpopt, which uses a smart allocation strategy to reduce the noise of lower order statistics, as discussed in \Cref{opt:noiseallocation}. \texttt{SF-2} doesn't support joins, so it is not reported.

We use $n_{corp}=2$ datasets: we fix cardinality $n=10K$ but vary join key domain size $d \in \{10,50,100,200\}$, then fix $d=100$ but vary $n \in \{100,5K,10K,50K\}$.
The results are shown in  \Cref{fig:join_betaerror_num_group} and \Cref{fig:join_betaerror_num_rows_per_group}, respectively.  \fdp, \fdpopt and \texttt{APM} have low $\beta$ error, while \texttt{TPM} and \texttt{SF} have high failure rates. \fdpopt outperforms \fdp due to better noise allocation.
\texttt{APM} outperforms \fdp and \fdpopt at large $d$ or small $n$ because \texttt{APM} adds noise directly to the aggregated statistics across join keys, resulting in a smaller amount of noise. In contrast, \fdp adds noise for each join key value. However, for large $n$, \fdp and \fdpopt outperform \texttt{APM} because it has high sensitivity due to high join fanouts~\cite{johnson2018towards}.

\Cref{fig:join_betaerror_num_users} and \Cref{fig:join_betaerror_requests} respectively vary the number of datasets and requests: $n_{corp}, n_{req}{\in}\{1, 5, 10, 50, 100\}$. 
Both \texttt{TPM} and \texttt{SF} have high failure rates, and \fdpopt outperforms \fdp.
\fdp and \fdpopt scale to arbitrary numbers of datasets and requests, while \texttt{APM} is restricted to $n_{corp}{\leq}5$ or $n_{req}{\leq}10$.

\begin{figure}
  \centering
    \begin{subfigure}{0.23\textwidth}
         \centering
         \includegraphics[width=\textwidth]{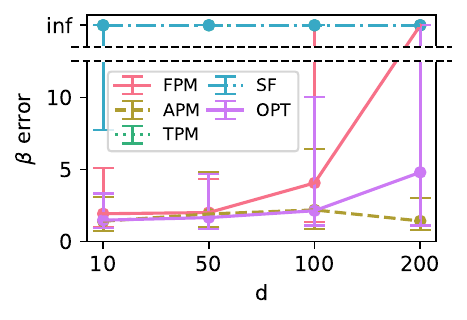}
         \vspace*{-8mm}
         \caption{$\beta$ error when varying $d$.}
    \label{fig:join_betaerror_num_group}
     \end{subfigure}
    \begin{subfigure}{0.23\textwidth}
         \centering
         \includegraphics[width=\textwidth]{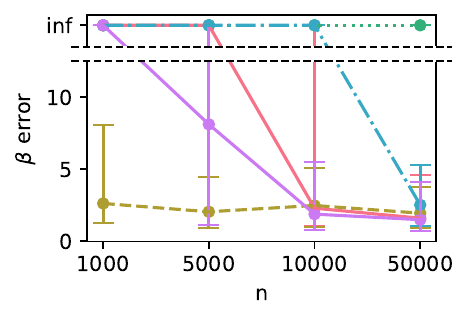}
         \vspace*{-8mm}
         \caption{$\beta$ error when varying $n$.}
         \label{fig:join_betaerror_num_rows_per_group}
     \end{subfigure}
    \begin{subfigure}{0.23\textwidth}
         \centering
         \includegraphics[width=\textwidth]{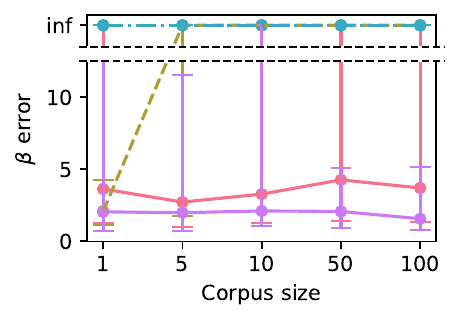}
         \vspace*{-7mm}
         \caption{$\beta$ error when varying $n_{corp}$.}
         \label{fig:join_betaerror_num_users}
     \end{subfigure}
    \begin{subfigure}{0.23\textwidth}
         \centering
         \includegraphics[width=\textwidth]{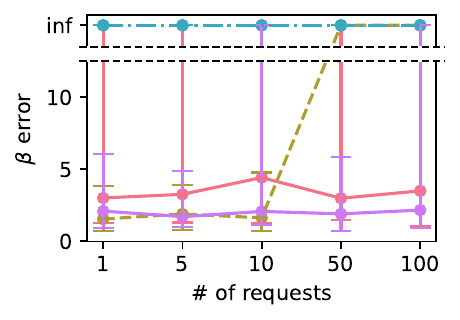}
         \vspace*{-7mm}
         \caption{$\beta$ error when varying $n_{req}$.}
         \label{fig:join_betaerror_requests}
     \end{subfigure}
     \vspace*{-4mm}
  \caption{Ablation tests for joins: (a,b). \fdp, \fdpopt and \texttt{APM} provide low error when varying $d$ and $n$, while \texttt{TPM} and \texttt{SF} are dominated by high noises (c,d). \fdp, \fdpopt show scalability for large repositories with numerous requests, whereas \texttt{APM} has high errors when $n_{corp}{>}5$ or $n_{req}{>}10$.}
  \vspace*{-6mm}
\end{figure}

\begin{figure}
  \centering
    \begin{subfigure}{0.26\textwidth}
         \centering
         \includegraphics[width=\textwidth]{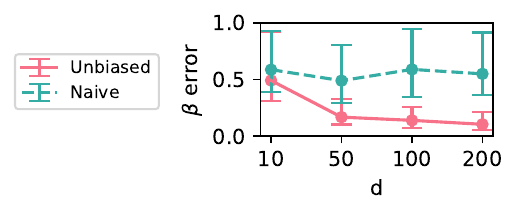}
         \vspace*{-7mm}
         \caption{$\beta$ error varying $d$.}
         \label{fig:betaerror_biased_group}
     \end{subfigure}
    \begin{subfigure}{0.17\textwidth}
         \centering
         \includegraphics[width=\textwidth]{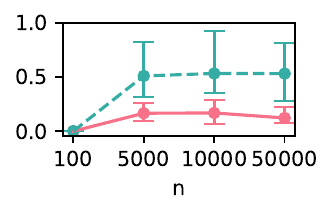}
         \vspace*{-7mm}
         \caption{$\beta$ error varying $n$.}
         \label{fig:betaerror_biased_num_rows_per_group}
     \end{subfigure}
     \vspace*{-5mm}
  \caption{$\beta$ error for naive and unbiased estimators over join. The unbiased estimator dominates the naive one.}
  \vspace*{-6mm}
\end{figure}

\vspace{-2mm}
\subsubsection{Join Unbiased estimator} 
Here, we compare the $\beta$ error of the unbiased estimator proposed in \Cref{sec:joinassumption} to the naive estimator over many-to-many joins. 
We first fix the number of tuples $n=10K$ but vary join key domain size $d \in \{10,50,100,200\}$, then fix $d=100$ but vary $n \in \{100,5K,10K,50K\}$, and report the results in \Cref{fig:betaerror_biased_group} and \Cref{fig:betaerror_biased_num_rows_per_group} respectively. 
As $d$ increases, the errors of the unbiased estimator converge to 0, while the biased estimator diverges as it fails to account for many-to-many join. When $n=100$, the naive estimator achieves similar performance, as each join key has only one tuple (so one-to-one join without bias). However, increasing $n$ introduces duplications and independence (for many-to-many join). The unbiased estimator reduces the noise and performs better than the naive estimator.

%% file: sections/nyc_exp_fig.tex
\begin{figure*}
  \centering
     \begin{subfigure}[t]{0.2\textwidth}
         \centering
         \includegraphics[width=\textwidth]{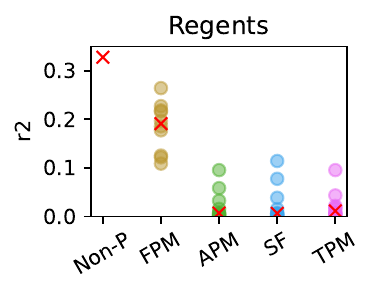}
\vspace*{-3mm}
     \end{subfigure}
     \begin{subfigure}[t]{0.18\textwidth}
         \centering
         \includegraphics[width=\textwidth]{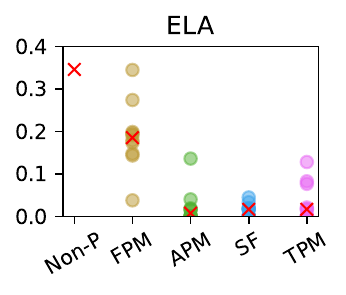}
\vspace*{-3mm}
     \end{subfigure}
     \begin{subfigure}[t]{0.18\textwidth}
         \centering
         \includegraphics[width=\textwidth]{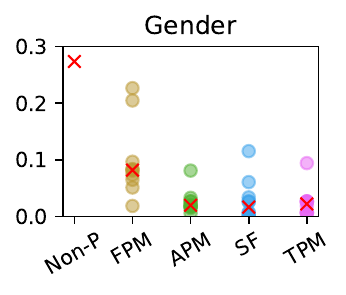}
\vspace*{-3mm}
     \end{subfigure}
     \begin{subfigure}[t]{0.18\textwidth}
         \centering
         \includegraphics[width=\textwidth]{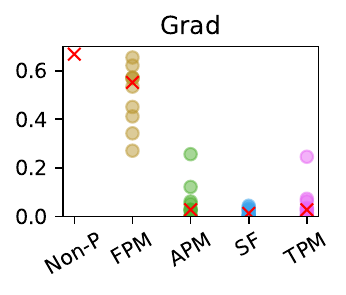}
\vspace*{-3mm}
     \end{subfigure}
     \begin{subfigure}[t]{0.18\textwidth}
         \centering
         \includegraphics[width=\textwidth]{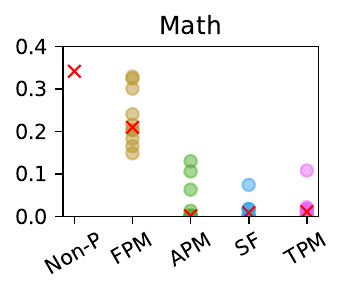}
\vspace*{-3mm}
     \end{subfigure}
     \vspace*{-4mm}
  \caption{Utility (non-privatized $r2$) of the returned combinations of datasets searched by different baselines over $10$ runs, with the median indicated by a red cross. \fdp exhibits significantly better utility than the other baselines.}
  \label{exp:nycutility}
\end{figure*}

\begin{figure*}
  \centering
     \begin{subfigure}[t]{0.263\textwidth}
         \centering
         \includegraphics[width=\textwidth]{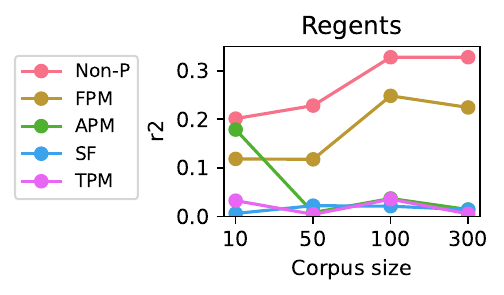}
\vspace*{-2mm}
     \end{subfigure}
     \begin{subfigure}[t]{0.18\textwidth}
         \centering
         \includegraphics[width=\textwidth]{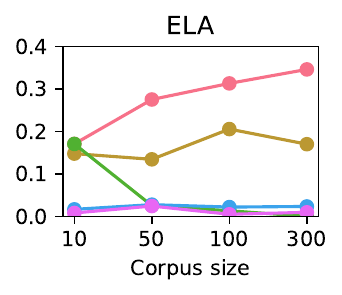}
\vspace*{-2mm}
     \end{subfigure}
     \begin{subfigure}[t]{0.18\textwidth}
         \centering
         \includegraphics[width=\textwidth]{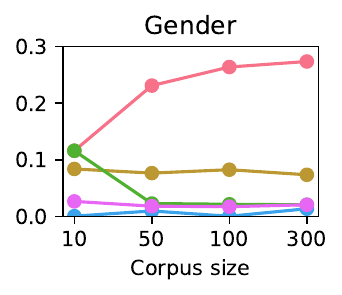}
\vspace*{-2mm}
     \end{subfigure}
     \begin{subfigure}[t]{0.18\textwidth}
         \centering
         \includegraphics[width=\textwidth]{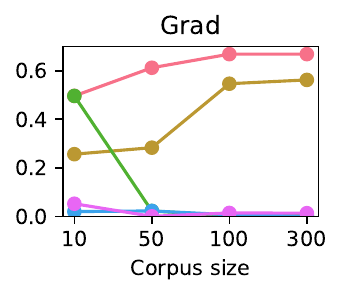}
\vspace*{-2mm}
     \end{subfigure}
     \begin{subfigure}[t]{0.18\textwidth}
         \centering
         \includegraphics[width=\textwidth]{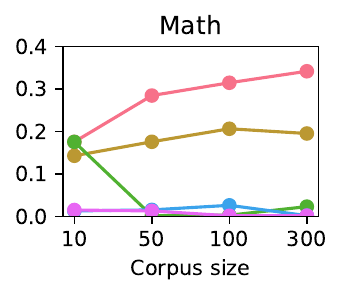}
         \vspace*{-2mm}
     \end{subfigure}
       \vspace*{-2mm}
  \caption{
  Utility (non-privatized $r2$) of the returned combinations of datasets searched by different baselines when varying the  corpus size $n_{corp}$. \fdp is scalable for large corpus, while \texttt{APM} only performs well for small $n_{corp}$.
   }

   \label{exp:nycrepo}
\end{figure*}

\begin{figure*}
  \centering
     \begin{subfigure}[t]{0.263\textwidth}
         \centering
         \includegraphics[width=\textwidth]{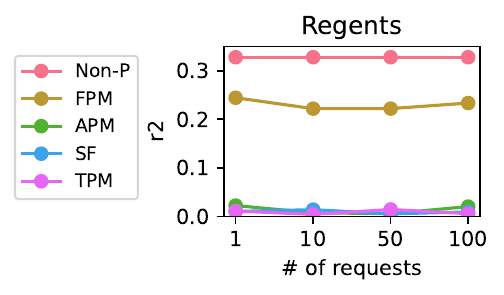}

     \end{subfigure}
     \begin{subfigure}[t]{0.18\textwidth}
         \centering
         \includegraphics[width=\textwidth]{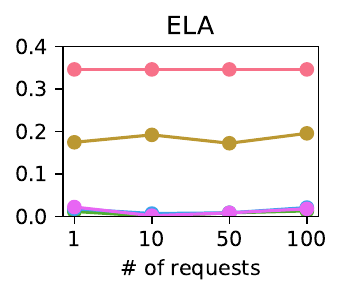}

     \end{subfigure}
     \begin{subfigure}[t]{0.18\textwidth}
         \centering
         \includegraphics[width=\textwidth]{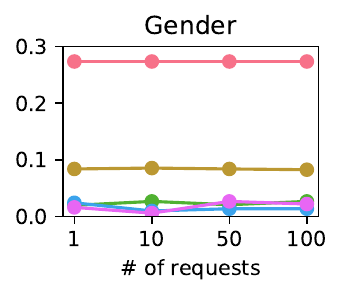}

     \end{subfigure}
     \begin{subfigure}[t]{0.18\textwidth}
         \centering
         \includegraphics[width=\textwidth]{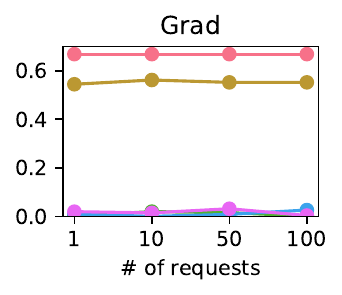}

     \end{subfigure}
     \begin{subfigure}[t]{0.18\textwidth}
         \centering
         \includegraphics[width=\textwidth]{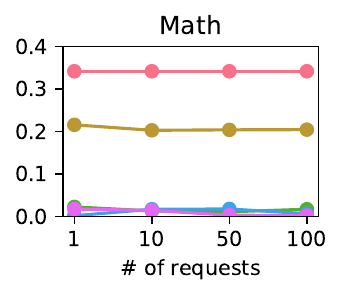}
     \end{subfigure}
  \caption{
  Utility (non-privatized $r2$) of the returned combinations of datasets searched by different baselines when varying the number of requests $n_{req}$.  \fdp consistently performs well, while the others either suffer from large noises (\texttt{TPM}, \texttt{SF}) or budget splits (\texttt{APM}).
  }
   \label{exp:nycrequest}
\end{figure*}

%% file: sections/union_exp_fig.tex
\begin{figure*}
  \centering
    \begin{subfigure}{0.23\textwidth}
         \centering
         \includegraphics[width=\textwidth]{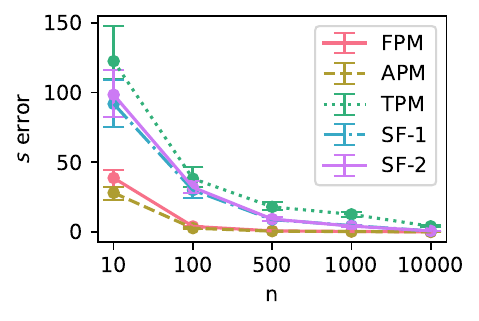}
         \vspace*{-5mm}
         \caption{$s$ error when varying $n$.}
         \label{fig:serror}
     \end{subfigure}
    \begin{subfigure}{0.23\textwidth}
         \centering
         \includegraphics[width=\textwidth]{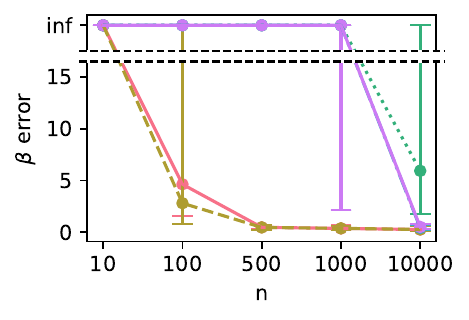}
          \vspace*{-5mm}
         \caption{$\beta$ error when varying $n$.}
         \label{fig:betaerror}
     \end{subfigure}
     \begin{subfigure}{0.23\textwidth}
         \centering
         \includegraphics[width=\textwidth]{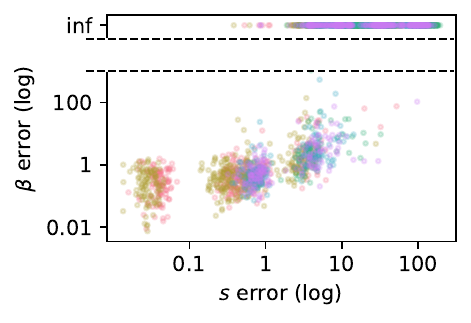}
          \vspace*{-5mm}
         \caption{Compare $\beta$ with $s$ error.}
         \label{fig:fs_relation}
     \end{subfigure}
    \begin{subfigure}{0.23\textwidth}
         \centering
         \includegraphics[width=\textwidth]{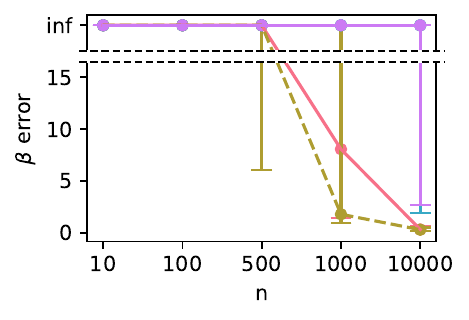}
          \vspace*{-5mm}
         \caption{$\beta$ error when $\epsilon{=}0.1$. }
         \label{fig:betaerror_small_eps}
     \end{subfigure}
     \begin{subfigure}{0.23\textwidth}
         \centering
         \includegraphics[width=\textwidth]{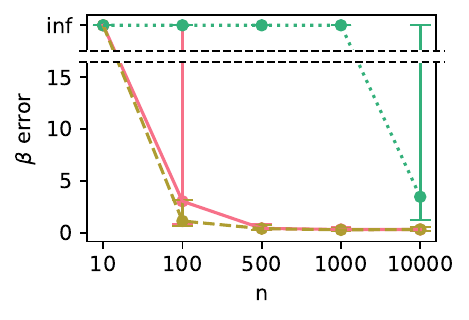}
          \vspace*{-5mm}
         \caption{$\beta$ error when $\delta{=}0$.}
         \label{fig:betaerror_pure}
     \end{subfigure}
    \begin{subfigure}{0.23\textwidth}
         \centering
         \includegraphics[width=\textwidth]{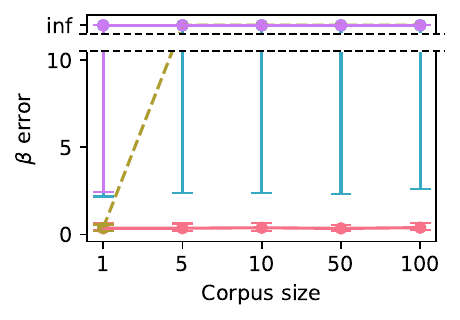}
          \vspace*{-5mm}
         \caption{$\beta$ error when varying $n_{corp}$.}
         \label{fig:betaerror_users}
     \end{subfigure}
     \begin{subfigure}{0.23\textwidth}
         \centering
         \includegraphics[width=\textwidth]{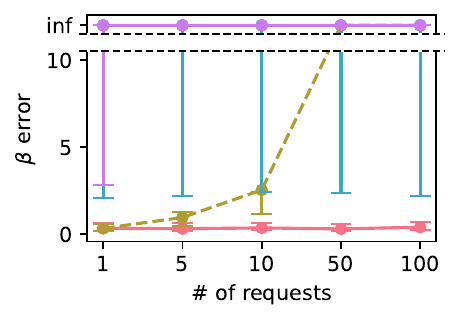}
          \vspace*{-5mm}
         \caption{$\beta$ error when varying $n_{req}$.}
         \label{fig:betaerror_requests}
     \end{subfigure}
     \begin{subfigure}{0.23\textwidth}
         \centering
         \includegraphics[width=\textwidth]{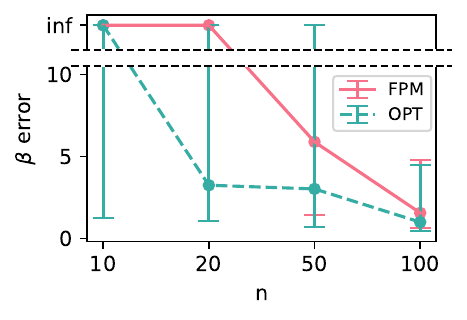}
          \vspace*{-5mm}
         \caption{$\beta$ error with optimized \fdp.}
         \label{fig:beta_error_unionopt}
     \end{subfigure}
   \vspace*{-3mm} 
  \caption{Ablation tests for unions: (a). (b). All baselines show reduced $s$ and $\beta$ errors as $n$ increases, with \fdp and \texttt{APM} exhibiting the least error; (c). Larger $s$ error results in higher $\beta$ errors and a greater risk of failure; (d). For $\epsilon=0.1$, the $\beta$ error plot shifts to the right; (e). For pure \texttt{DP} with $\delta=0$,  \fdp can adjust to it and offer comparable utility, while \texttt{SF} fails; (f). (g). \fdp is scalable for large corpora with large numbers of requests, while \texttt{APM} deteriorates significantly when $n_{corp} > 5$ or $n_{req} > 10$; (h). For simple linear regression when only $\beta_x$ is of interest,  \fdpopt reduces the $\beta$ errors and failure risk, particularly for small $n$.}
\end{figure*}

%% file: sections/related.tex
\vspace{-3mm}
\section{Related Works}
\label{sec:related}
\vspace{-1mm}
\stitle{Dataset search.}
Traditional data discovery focuses on augmentable (i.e., joinable or unionable) datasets~\cite{castelo2021auctus,fernandez2018aurum}, whereas
recent dataset search platforms~\cite{santos2022sketch,chepurko2020arda,nargesian2022responsible,li2021data} are based on data augmentation {\it for ML tasks}.
However, none addresses privacy concerns.

\stitle{Differential Privacy for Databases.}
Differentially private databases can query over multiple tables~\cite{johnson2018towards, wilson2019differentially, kotsogiannis2019privatesql}. They apply \texttt{DP} mechanisms to query results over joins and unions. Notably, join poses a significant  \texttt{DP} challenge due to the exponential sensitivity growth along the join path. \fdp may offer a solution by decomposing join query into smaller, bounded-sensitivity statistics.

\stitle{Federated ML.} 
These methods let each untrusted party compute and privatize their local gradients for horizontal~\cite{wei2020federated,shokri2015privacy,truex2020ldp,zhao2020local} or vertical~\cite{wang2020hybrid,hardy2017private} federated ML, which are then combined to train the final model. However, the gradient is specific to training a single model. In contrast, data search repeatedly trains new models to evaluate candidate augmentations, requiring budget splits.

\stitle{Differentially Private Sufficient Statistics.} 
Previous works use sufficient statistics~\cite{huggins2017pass} for generalized linear models and apply perturbations~\cite{kulkarni2021differentially} to guarantee  \texttt{DP}. For linear regression, sufficient statistics perturbation, particularly with regularization, outperforms other \gdp mechanisms including objective perturbation and noisy SGD~\cite{wang2018revisiting, NEURIPS2022_5bc3356e}.
However, they only consider ML on a single dataset.

\stitle{Factorized ML.}
Factorized ML decomposes ML models into semi-ring queries, designs algebraic operators to combine them, and achieves asymptotically lower time complexity.  They support models like ridge regression~\cite{schleich2016learning}, random forests~\cite{joinboost}, SVM~\cite{khamis2020functional}, and factorization machine~\cite{schleich2019layered}.  
None are differentially private.
\revise{We are the first to apply \texttt{DP} to factorized linear regression. Future work aims to extend \sys to other proxy models like random forests used in prior data search~\cite{chepurko2020arda}. The challenge lies in the lack of a closed-form solution in random forests,  requiring iterative computation of semi-ring aggregates based on tree splits, which are not reusable and costly to privatize. 
To improve search utility, we plan to explore  (1) alternative trust models where requesters trust the platform to lessen  noise, and (2) differentially private synopses~\cite{tantipongpipat2021differentially, yoon2018pategan} based on monomial semi-ring to generate synthetic data, allowing the training of arbitrarily  complex models as post-processing.}

\vspace*{-2mm}

%% file: sections/conclusions.tex
\vspace{-2mm}
\section{Conclusions}
\sys is a differentially private data search platform that searches large corpora to find datasets to improve ML performance via augmentation.
\sys employs \fdp, a novel mechanism that privatizes sufficient semi-ring statistics, which can be reused without incurring additional \texttt{DP} cost. 
In a deep study of linear regression, we propose an unbiased estimator for many-to-many joins, prove parameter bounds under augmentations, and propose an optimization to allocate \texttt{DP} budget better.
On a ${>}300$ dataset corpus, \fdp achieves an $r2$ score (${\sim}50{-}90\%)$ of non-private search, while other mechanisms (\texttt{TPM}, \texttt{APM}, \texttt{SF}) report negligible $r2$ scores ${<}0.02$.

%% file: sections/fdpproof.tex
\section{\fdp Sensitivity}
\label{sensitivity}

For union, query $q_u: \mathcal{D}^n \longrightarrow \mathcal{S}$ where datasets in $\mathcal{D}^n$ contains $m$ features and $\mathcal{S} = \{v \ | \ v \in \mathbb{R}^{m^i}, \ i \in \{0, \dots, k\}\}$. $q_u$ 
returns a set of vectors $s \in \mathcal{S}$ containing the sum of $k$-order monomial semi-ring across all tuples. For analysis convenience, we will overload the notation a bit and treat $s$ as a single vector by concatenating $\{v\}_{v \in s}$. Let $t_1, t_2 \in \mathcal{D}^n$ and $t^i_1, t^i_2$ be vectors of $i$-order monomial with respect to $t_1[f_1], \dots, t_1[f_m]$ and $t_2[f_1], \dots, t_2[f_m]$. Let $\sigma(k)$ denote the set of series $[k_1, \dots, k_m]$ such that $k_i \in \mathbb{N}$ and $\sum k_i = k$. The squared distance between $t^k_1$ and $t^k_2$ can be computed as
\begin{align*}
    &\quad \sum_{\substack{[k_1, \dots, k_m] \\ \in \sigma(k)}}(\prod_{i=1}^mt_1[f_i]^{k_i} - \prod_{i=1}^mt_2[f_i]^{k_i})^2\\
    &\le \sum_{\substack{[k_1, \dots, k_m] \\ \in \sigma(k)}}{k \choose k_1, \dots, k_m}(\prod_{i=1}^mt_1[f_i]^{k_i} - \prod_{i=1}^mt_2[f_i]^{k_i})^2
\end{align*}

By multinomial theorem, we may rewrite the last equation as 
\begin{align*}
    &\quad (\sum_{i=1}^m t_1[f_i]^2)^k + (\sum_{i=1}^m t_2[f_i]^2)^k - 2(\sum_{i=1}^m t_1[f_i]t_2[f_i])^k\\
    &\le 2B^{2k} -  2(\sum_{i=1}^m t_1[f_i]t_2[f_i])^k
\end{align*}

That is, when $k$ is even, the latter term is strictly positive. Hence $\|t_1^k - t_2^k\|^2 \le 2B^{2k}$. Let $D_1, D_2 \in \calD^n$ be two neighbouring datasets differ in one tuple, $t_1$ and $t_2$. The sensitivity of $q_u(\cdot)$ can be computed as 
$$\Delta_{q_u} = \max\|q_u(D_1) - q_u(D_2)\| \le \sqrt{\sum_{i=1}^k 1_{2\mathbb{Z} + 1}(i)4B^{2i} + 1_{2\mathbb{Z}}(i)2B^{2i}}$$

For join, we inherit the notations from the union case and let $\ell$ be the number of join keys. $q_j: \mathcal{D}^n \longrightarrow \mathcal{S}$ returns a set of vectors $s \in \mathcal{S}$ where $s$ concatenates vectors returned by $q_u$ on each partition of tuples for each join key. Consider two cases: (1) $t_1$ and $t_2$ have the same join key (2) $t_1$ and $t_2$ have different join keys. In the former case,
\begin{align*}
    \Delta_{q_j} &= \max\|q_j(D_1) - q_j(D_2)\|_2\\
    &\le \sqrt{\sum_{i=1}^k 1_{2\mathbb{Z} + 1}(i)4B^{2i} + 1_{2\mathbb{Z}}(i)2B^{2i}}
\end{align*}

In the latter case, 
\begin{align*}
    \Delta_{q_j} &= \max\|q_j(D_1) - q_j(D_2)\|_2\\
    &= \max\sqrt{\sum_i^k \|t_1^i\|^2 + \sum_i^k \|t_2^i\|^2}\\
    &\le \sqrt{\max\sum_i^k \|t_1^i\|^2 + \max\sum_i^k \|t_2^i\|^2}\\
    &\le \sqrt{2\sum_{i=0}^k B^{2i}}
\end{align*}

Hence, $\Delta_{q_j} = \max(\sqrt{\sum_{i=1}^k 1_{2\mathbb{Z} + 1}(i)4B^{2i} + 1_{2\mathbb{Z}}(i)2B^{2i}}, \sqrt{2\sum_{i=0}^k B^{2i}})$

%% file: sections/erroranalysis.tex
\section{Error Analysis}
\label{utility}
\begin{table}[t!]
    \centering
    \caption{Notation}\label{table:notation} \vspace{-3mm}
    \begin{tabular}{r p{6.5cm}}
        \textbf{Notation} & \textbf{Description}\\ \hline
        
        $R_i$ & relations of providers/requesters.\\

        $n$ & size of each relation.\\
        
        $J, Dom(J), d$ & join key, domain of join key, domain size.\\
        
        $B$ & the $\ell_2$ distance upper bound of each tuple in each relation.\\

        $\var, \cov$ & the empirical estimation of the variance and covariance.\\

        $\pvar, \pcov$ & the privatized empirical estimation of the variance and covariance.\\

        $B_1, B_2$ & $1-p$ confidence bound on $|\var - \pvar|$ and $|\cov - \pcov|$.\\
        \hline
    \end{tabular} 
\end{table}
For a single data provider, with relation $R$ where $|R|=n$ and $i$ be any integer from 1 to $k$. \ldp computes $[f, f^2, \dots, f^k]$ for each tuple and adds noise to each of them. By similar analysis to that of \fdp, \ldp's sensitivity is the same as $\Delta$ in \fdp for both union and join. Hence, for each tuple, $t$, from $R$, $t[\widetilde{f^i}] \sim t[f^i] + \calN(0, (2\ln(1.25/\delta)\Delta^2/\epsilon^2)$.
The empirical expectation of $f^i$ can be computed as $$\widetilde{f_\ldp^i} = \frac{1}{n} \left(\sum_{t \in R} t[f^i]\right) + e_i \quad e_i \sim \calN(0, 2\ln(1.25/\delta)\Delta^2/n\epsilon^2)$$
Putting everything together, and by the assumption that $k$ is a small constant, we have

\begin{align*}
    E[\|s'_\ldp - \hat{s}\|] &= E\left[\sqrt{\sum_{i=1}^k \left(\widetilde{f_\ldp^i} - \frac{1}{n} \left(\sum_{t \in R} t[f^i]\right)\right)^2}\right]\\
    &\le \sqrt{E\left[\sum_{i=1}^k \left(\widetilde{f_\ldp^i} - \frac{1}{n} \left(\sum_{t\in R} t[f^i]\right)\right)^2\right]}\\
    &= \sqrt{\sum_{i=1}^k E[e_i^2]}\\
    &= O(\Delta/\sqrt{n}\epsilon)
\end{align*}


For \fdp, the only difference is that 
$$\widetilde{f_\fdp^i} \sim \frac{1}{n} \left(\sum_{t \in R} t[f^i] \right) + e_i \quad e_i \sim \calN(0, 2\ln(1.25/\delta)\Delta^2/n^2\epsilon^2)$$

Following the same line of derivation, 
$$E[\|\widetilde{s_\fdp} - \hat{s}\|] = O(\Delta/n\epsilon)$$

However, \gdp needs to account for any possible combination of a single buyer and a subset of sellers, where each party's privacy needs to be protected. Specifically, each buyer appears in $2^{n_{corp}} - 1$ combinations, since each buyer requires at least one seller. On the other hand, for a fixed buyer, each seller is involved in $2^{n_{corp}-1}$ combinations. Hence, each seller will appear in $n_{req}2^{n_{corp}-1}$ combinations in total. Because each seller and buyer have privacy budget $(\epsilon, \delta)$, in order to provide privacy guarantees for each party in any combination, the amount of privacy budget spent on perturbing pre-normalized $s$ is $\epsilon' = \min(\epsilon/(2^{n_{corp}} - 1), \epsilon/n_{req}2^{n_{corp}-1})$ and $\delta' = \min(\delta/(2^{n_{corp}} - 1), \delta/n_{req}2^{n_{corp}-1})$ 

$$\widetilde{f^i_\gdp} = \frac{1}{n} \left(\sum_{t\in R} t[f^i] \right) + e_i \quad e_i \sim \calN(0, 2\ln(1.25/\delta')\Delta^2/n^2\epsilon'^2)$$

Based on the same line of analysis above
$$E[\|\widetilde{s_\gdp} - \hat{s}\|] = O(n_{req}2^{n_{corp} - 1}\Delta/n\epsilon)$$

Now we consider \texttt{SF-1}, based on \cite{erlingsson2019amplification}, it suffice to guarantee $\epsilon/\sqrt{n}$-DP for local responses to achieve $(\epsilon, \delta)$-DP from the central's perspective, where each tuple $t$ in $R$ satisfies $t[\widetilde{f^i}] \sim t[f^i] + \Lap(\Delta/\sqrt{n}\epsilon)$. Then we have $$\widetilde{f_{\texttt{SF-1}}^i} = \frac{1}{n} \sum_{t \in R} (t[f^i]  + e_{t}) \quad e_{t} \sim \Lap(\Delta/\sqrt{n}\epsilon)$$

Then, we have
\begin{align*}
    E[\|\widetilde{\texttt{SF-1}} - \hat{s}\|] &= E\left[\sqrt{\sum_{i=1}^k \left(\widetilde{f_{\texttt{SF-1}}^i} - \frac{1}{n} \left(\sum_{t \in R} t[f^i]\right)\right)^2}\right]\\
    &\le \sqrt{\sum_{i=1}^k E\left[\left(\frac{1}{n} \sum_{t \in R} e_{t}\right)^2\right]}
\end{align*}

Since $E\left[\left(\frac{1}{n} \sum_{t \in R} e_{t}\right)\right] = 0$, it follows that $$E\left[\left(\frac{1}{n} \sum_{t \in R} e_{t}\right)^2\right] = Var\left(\frac{1}{n} \sum_{t \in R} e_{t}\right) = \frac{\Delta^2}{n^2\epsilon^2}$$

Substituting back to the equation, and based on assumption that $k$ is small, we have $$E[\|\widetilde{s_{\texttt{SF-1}}} - \hat{s}\|] = O(\Delta/n\epsilon)$$

For \sdp-2, just like \gdp, it also needs to account for all possible combination of a single buyer and any subsets of sellers in the centralized shuffler. However, the differences are that \texttt{SF-2} allows each combination's privacy guarantee to be amplified by an amount of $O(\sqrt{n})$, and that \texttt{SF-2} draw random noises from Laplace distribution instead of Gaussian distribution. That is,

$$\widetilde{f_{\texttt{SF-2}}^i} = \frac{1}{n} \sum_{t \in R} (t[f^i]  + e_{t}) \quad e_{t} \sim \Lap(\Delta/\sqrt{n}\epsilon')$$

Hence, the expected utility can be computed following the same line of derivation of \texttt{SF-1}. That is 

$$E[\|\widetilde{s_{\texttt{SF-2}}} - \hat{s}\|] = O(\Delta/n\epsilon') = O(n_{req}2^{n_{corp}-1}\Delta/n\epsilon)$$


%% file: sections/unbiasedproof.tex
\section{Unbiased proof}
\label{sec:unbiasedproof}
We make the simplifying assumption that $J$ is uniformly distributed: if $d=|dom(J)|$, then each $j \in J$ appears $n/d$ times in $R$. Moreover, the projection operator $\pi$ will not remove duplicates in $R$ so $|\pi_{J, f_1}(R)| = |\pi_{J, f_2}(R)| = n$.

\begin{prop}[Expected $s$ over $R_\Join$] Assume that $R_\Join$ is the population. For any other 1,2-order monomial $p$,
\begin{equation*}
E[p] = s[p]/s[c]
\end{equation*}
where c is the count (0-order monomial). Then $E[p]$ is the expected  $s$ over $R_\Join$.
\end{prop}

\begin{prop}[Unbiased Estimator of $s$ over $R$]
\begin{equation*}
\widehat{E[p]} = \begin{cases}
\mathrlap{f_1f_2 = \frac{1-n}{1-d} \frac{s[f_1f_2]}{s[c]} + \frac{n-d}{1-d} \frac{s[f_1]}{s[c]} \frac{s[f_2]}{s[c]}}\\
&\text{for $f_1 \in F_1, f_2 \in F_2$}\\
p = s[p]/s[c] &\text{for any other monomial $p$}\\
\end{cases}
\end{equation*}
$\hat{s}$ is an unbiased estimator of $s$.
\end{prop}

\begin{proof}
We demonstrate that, for any 1,2-order monomial where features are from the same relation, $E[s[p]/s[c]] = E[p]$. 
\begin{align*}
s[c] &= n \cdot n/d \\
E[s[f]/s[c]] &= E[(\sum_{t \in R} t[f] \cdot n/d)/(n \cdot n/d )] = \sum_{t \in R} E[t[f]] /n = E[f] \\
E[s[f_1 f_2]/s[c]] &= E[(\sum_{t \in R} t[f_1]\cdot t[f_2] \cdot n/d)/(n \cdot n/d )] \\
& = \sum_{t \in R} E[t[f_1]\cdot t[f_2]] /n = E[f_1 f_2]    
\end{align*}
The first equality is because for each join key, the cartesian product is computed, leading to duplication of tuples with the same join key in both tables by $n/d$ times. The count is also increased by $n/d$, thus resulting in the equality $s[p]/s[c] = E[p]$.

However, this equality does not hold for the  $f_1 f_2$, where $f_1$ and $f_2$ are from different relations. In this case, $f_1$ from $R_1$ is paired with all $f_2$ from $R_2$ with the same join key, but the information about which $f_2$ is paired with $f_1$ in original $R$ is lost. Nonetheless, we can still estimate $E[f_1 f_2]$ by exploiting the covariance across groups.

We first analyze $E[f_1 f_2]$ for a single join key value $j$. We use notation $s^j$ to denote the monomial semi-ring for the join key value $k$. Consider random variable of the average:
\begin{align*}
s^j_1 =& s^j[f_1]/s^j[c] = \left(\sum_{t \in \sigma_{j}(R)}t[f_1]
\cdot n/d\right)/(n/d)^2 = \sum_{t \in \sigma_{j}(R)}\frac{t[f_1]}{n/d}\\
s^j_2 =& s^j[f_2]/s^j[c] = \left(\sum_{t \in \sigma_{j}(R)}t[f_2]\cdot n/d\right)/(n/d)^2 = \sum_{t \in \sigma_{j}(R)}\frac{t[f_2]}{n/d}
\end{align*}
$s^j_1$ and $s^j_2$ can be understood as the mean of $f_1$ and $f_2$ from the sample $\sigma_{j}(R)$. It is obvious that $E[s^j_1] = E[f_1]$ and $E[s^j_2] = E[f_2]$.

From the definition of covariance, we have: 
\begin{align*}
E[s^j_1 s^j_2]&= cov(s^j_1,s^j_2) + E[s^j_1]E[s^j_2] \\
&= cov\left(\sum_{t \in \sigma_{j}(R)}\frac{t[f_1]}{n/d},\sum_{t \in \sigma_{j}(R)}\frac{t[f_2]}{n/d}\right) + E[f_1]E[f_2]
\end{align*}

We next compute the $cov$:
\begin{align*}
cov\left(\sum_{t \in \sigma_{j}(R)}\frac{t[f_1]}{n/d},\sum_{t \in \sigma_{j}(R)}\frac{t[f_2]}{n/d}\right)&= \frac{d^2}{n^2}\sum_{\substack{t_1 \in \sigma_{j}(R)\\t_2 \in \sigma_{j}(R)}}cov(t_1[f_1],t_2[f_2])\\
&=\frac{d^2}{n^2}\sum_{t \in \sigma_{j}(R)}cov(t[f_1],t[f_2])\\
&=\frac{d}{n}cov(f_1,f_2)
\end{align*}
The first line is by the property of covariance and the second line is by the independence between tuples. Therefore,
\begin{align*}
E[s^j_1 s^j_2]&= \frac{d}{n}cov(f_1,f_2) + E[f_1]E[f_2]
\end{align*}
Next, consider the random variables across join keys:
\begin{align*}
s_1 =& s[f_1]/s[c] = \sum_{t \in R}t[f_1]/n\\
s_2 =& s[f_2]/s[c] = \sum_{t \in R}t[f_2]/n\\
s_{1,2} =& s[f_1f_2]/s[c] = \sum_{j\in dom(J)} s^j_1\cdot s^j_2/d
\end{align*}
where $s_1$ and $s_2$ are the average across join keys. $s_{1,2}$ is the average products across join keys. It is obvious that $E[s_1] = E[f_1],E[s_2] = E[f_2]$.
We next study $E[s_1s_2]$ and $E[s_{1,2}]$:
\begin{align*}
E[s_1 s_2]&= cov(s_1,s_2) + E[s^j_1]E[s^j_2] \\
&= cov\left(\sum_{t \in R}t[f_1]/n,\sum_{t \in R}t[f_2]/n\right) + E[f_1]E[f_2]
\end{align*}

Similar as before,
\begin{align*}
cov\left(\sum_{t \in R}t[f_1]/n,\sum_{t \in R}t[f_2]/n\right)&= \frac{1}{n^2}\sum_{\substack{t_1 \in R\\t_2 \in R}}cov(t_1[f_1],t_2[f_2])\\
&=\frac{1}{n^2}\sum_{t \in R}cov(t[f_1],t[f_2])
=\frac{1}{n}cov(f_1,f_2)
\end{align*}

Therefore: 
\begin{align*}
E[s_1 s_2]&= \frac{1}{n}cov(f_1,f_2) + E[f_1]E[f_2]
\end{align*}

Finally,
\begin{align*}
E[s_{1,2}]&= \sum_{j\in dom(J)} E[s^j_1\cdot s^j_2]/d = \frac{d}{n}cov(f_1,f_2) + E[f_1]E[f_2]
\end{align*}

Putting everything together, we show that 
$\frac{1-n}{1-d} s_{1,2} + \frac{n-d}{1-d} s_1\cdot s_2$ is an unbiased estimator of $E[f_1f_2]$:
\begin{align*}
E[\frac{1-n}{1-d} s_{1,2} + \frac{n-d}{1-d} s_1\cdot s_2] =& \frac{1-n}{1-d} E[s_{1,2}] + \frac{n-d}{1-d} E[s_1s_2]\\
=& \frac{1-n}{1-d}(\frac{d}{n}cov(f_1,f_2) + E[f_1]E[f_2])+ \\
&\frac{n-d}{1-d}(\frac{1}{n}cov(f_1,f_2) + E[f_1]E[f_2])\\
=& cov(f_1,f_2) + E[f_1]E[f_2] = E[f_1f_2]
\end{align*}
The first line is by the linearity of expectation, and the last line is by the definition of covariance.

\end{proof}

%% file: sections/linear_regression_proof.tex
\section{Confidence Bound of linear regression}
\label{app:linear}

Let $\sigma=\sqrt{2\ln(1.25/\delta)}\Delta/\epsilon$ where $\Delta = O(B^2)$. We are interested in $n, B \rightarrow \infty$ and $\epsilon, \delta, p \rightarrow 0$ in our analysis. Hence $\sigma = O\left(\frac{B^2\sqrt{\ln(1/\delta)}}{\epsilon}\right)$. The privatized empirical expectation of the moments are defined as $\widetilde{E[X]} = \widehat{E[X]} + e_1, \widetilde{E[X^2]} = \widehat{E[X^2]} + e_2, \widetilde{E[XY]} = \widehat{E[XY]} + e_3$ and $\widetilde{E[Y]} = \widehat{E[Y]} + e_4$. Then, we have $e_1, e_2, e_3, e_4 \sim \mathcal{N}(0,\sigma^2 / n^2)$
\begin{lemma}[High-probability bound on $\widetilde{\sigma_x}$]
\label{lem:varxbound}
Given $\var = \widehat{E[X^2]} - \widehat{E[X]}^2$ and $\pvar = \widetilde{E[X^2]} - \widetilde{E[X]}^2$, with probability at least $1-p$, $|\var - \pvar| = O(B_1)$ where $$B_1 = \frac{B^4\ln(1/\delta)\ln(1/p)}{\epsilon^2 n}$$
\end{lemma}

\begin{proof}
By assumption that each tuple's $\ell_2$ norm is bounded by $B$, each feature must also be bounded by $B$. Based on Gaussian tail bound, with probability at least $1-p/4$, $|e_i| \le \sigma\sqrt{2\ln(8/p)}/n$.
\begin{align*}
    |\var - \pvar| &= |e_1 -2e_2\sum x/n -e_2^2|\\
    &\le |e_1| + |2e_2\sum x/n| + |e_2^2|\\
    &\le \frac{\sigma\sqrt{2\ln(8/p)}}{n}\left(1 + 2B + \frac{\sigma\sqrt{2\ln(8/p)}}{n}\right)\\
    &=O\left(\frac{B^2\sqrt{\ln(1/\delta)\ln(1/p)}}{\epsilon n} + \frac{B^2\ln(1/\delta)\ln(1/p)}{\epsilon^2n^2}\right)\\
    &=O\left(\frac{B^4\ln(1/\delta)\ln(1/p)}{\epsilon^2 n}\right)
\end{align*}
\end{proof}

Similarly, \ref{lem:varxbound} can be used to derive the high probability bound on $\widetilde{\sigma^2_{xy}}$, that is $$|\cov - \pcov| \le B_2 = O\left(\frac{B^4\ln(1/\delta)\ln(1/p)}{\epsilon^2 n}\right)$$ 
Since the condition to satisfy both bounds coincide, with probability at least $1-p$, $|\widehat{\sigma^2_{x}} - \widetilde{\sigma^2_{x}}| \le B_1$ and $|\widehat{\sigma^2_{xy}} - \widetilde{\sigma^2_{xy}}| \le B_2$.

Let $$\tau_1 = B_1/\widehat{\sigma_x^2} = O\left(\frac{B^4\ln(1/\delta)\ln(1/p)}{\epsilon^2 n \var}\right) = \tau_2$$
When $\tau_1, \tau_2 < 1$ and with probability at least $1-p$, we may prove \ref{thm:pfactorized} as
\begin{proof}
\begin{align*}
    |\hat{\beta_x} - \tilde{\beta}_x| &= \left|\frac{\cov}{\var} - \frac{\pcov}{\pvar}\right| = \left|\frac{\cov}{\var} - \frac{\pcov}{\var} + \frac{\pcov}{\var} - \frac{\pcov}{\pvar}\right|\\
    &\leq \frac{|\pcov - \cov|}{\var} + \pcov\cdot \left|\pvar^{-1} - \var^{-1}\right|\\ \label{line:inverse}
    & = \frac{|\pcov - \cov|}{\var} + \pcov\cdot \frac{|\pvar - \var|}{\pvar \var}\\
    &\leq \tau_2 + (\cov + \tau_2\cdot\var)\frac{\tau_1\var}{(1-\tau_1)(\var)^2} \\
    &= \tau_2 + \frac{\tau_1}{1-\tau_1}\left(\frac{\cov}{\var} + \tau_2\right) = \tau_2 + \frac{\tau_1}{1-\tau_1}\left(\hat{\beta_x} + \tau_2\right)
\end{align*}
\end{proof}

\stitle{Extension to Factorized ML.}
The confidence bounds can be extended for factorized ML. The difference boils down to $B_1$, and the rest are the same. For union, let $R = R_1 \cup R_2 ... \cup R_k$ where $|R_i| = n$ and $k \rightarrow \infty$. Then, for $e_i \sim \calN(0, \sigma^2)$,

\begin{align*}
\widetilde{E[x^2]} = \frac{\sum_i^k(\sum x^2 + e_i)}{kn} \sim \widehat{E[X^2]} + \calN(0, \sigma^2/kn^2)
\end{align*}
Therefore, with probability at least $1-p/4$, $|\widetilde{E[X^2]} - \widehat{E[X^2]}| \le \sigma\sqrt{2\ln(8/p)}/\sqrt{k}n = O\left(\frac{B^2\sqrt{\ln(1/\delta)\ln(1/p)}}{\epsilon \sqrt{k}n}\right)$ (same for all other 3 moments $E[X], E[Y], E[XY]$), by minor changes in \ref{lem:varxbound}, yielding new bounds on $\tau_1$ and $\tau_2$ as 
\begin{align*}
    \tau_1 = B_1/\widehat{\sigma_x^2} &= O\left(\frac{B^2\sqrt{\ln(1/\delta)\ln(1/p)}}{\epsilon \sqrt{k}n\var} + \frac{B^4\ln(1/\delta)\ln(1/p)}{\epsilon^2kn^2\var}\right)\\
    &= O\left(\frac{B^4\ln(1/\delta)\ln(1/p)}{\epsilon^2 \sqrt{k}n\var}\right) = \tau_2
\end{align*}

For join, consider $R[x,y,J] = R_1[x,J] \Join R_2[y,J]$ and $d = |dom(J)|$ where $d \rightarrow n$. In contrast to union, there is additional noise added to the zero-th moment of each join key. i.e. the count of tuples within each join key. To avoid the scenario where this number is non-positive, an additional assumption is required~\cite{wang2018revisiting} that the noise is bounded by $o(n/d)$. Note that in the unbiased estimation, the privatized $s[c]$ is computed as 
\begin{align*}
    \widetilde{s[c]} &= \sum_{i \in J} (n/d +o(n/d))(n/d +o(n/d))\\ 
    &= \sum_{i \in J} (n/d +o(n/d))^2\\ 
    &= d(n/d +o(n/d))^2
\end{align*}

Then, for $e_{i, 1}, e_{i, 2}, e_{i, 3} \sim \calN(0, \sigma^2)$ defined as the Gaussian noise added to $\sum_{t \in R_1.i} x^2, \sum_{t \in R_1.i} x, \sum_{t \in R_2.i} y$ for each join key $i \in J$, with probability at least $1 - p/4$, $\sum_{i \in J}e_{i, j} \sim \calN(0, d\sigma^2)$ and $\sum_{i \in J}e_{i, j} \le \sigma\sqrt{2d\ln(8/p)} = O\left(\frac{B^2\sqrt{d\ln(1/\delta)\ln(1/p)}}{\epsilon}\right)$ for $j = \{1,2, 3\}$

\begin{align*}
\widehat{E[X^2]} &= \frac{\sum_{i\in J} (\sum_{t \in R_1.i} x^2)
n/d}{n^2/d} = \frac{\sum x^2}{n}\\
\widetilde{E[X^2]} &= \frac{\sum_{i\in J} \left((\sum_{t \in R_1.i} x^2) + e_{i, 1}\right)\left(n/d + o(n/d)\right)}{\sum_{j \in J} (n/d + o(n/d))(n/d + o(n/d))}
\end{align*}

By expanding $\widetilde{E[X^2]}$, we have
\begin{align*}
\widetilde{E[X^2]} &=\frac{(n/d)\sum x^2 + (n/d) \cdot \sum_{i \in J}e_{i, 1} + o(n/d)\sum x^2 + o(n/d)\sum_{i \in J}e_{i, 1}}{n^2/d + 2n\cdot o(n/d) + d \cdot o(n^2/d^2)}\\
&=\frac{(\frac{\sum x^2}{n} + (\sum_{i \in J}e_{i, 1})/n)(1  + o(n/d)(d/n))}{1 + 2(d/n)\cdot o(n/d) + (d^2/n^2) \cdot o(n^2/d^2)}\\
&=\frac{\frac{\sum x^2}{n} + (\sum_{i \in J}e_{i, 1})/n + o(1)(\frac{\sum x^2}{n} + (\sum_{i \in J}e_{i, 1})/n)}{1 + o(1)}\\
&=(1+o(1))\left(\frac{\sum x^2}{n} + (\sum_{i \in J}e_{i, 1})/n + o(1)(\frac{\sum x^2}{n} + (\sum_{i \in J}e_{i, 1})/n)\right)
\end{align*}

Hence
\begin{align*}
    |\widetilde{E[X^2]} - \widehat{E[X^2]}| &= O\left(\frac{\sum_{i \in J}e_{i, 1}}{n}\right)\\
    &=O\left(\frac{B^2\sqrt{d\ln(1/\delta)\ln(1/p)}}{\epsilon n}\right)
\end{align*}
Similarly, and based on $d \rightarrow n$
\begin{align*}
    |\widetilde{E[X]}^2 - \widehat{E[X]}^2| &= O\left(2\left(\frac{\sum x}{n}\right)\left(\frac{\sum_{i \in J}e_{i, 1}}{n}\right) + \left(\frac{\sum_{i \in J}e_{i, 1}}{n}\right)^2\right)\\
    &=O\left(\frac{B^3\sqrt{d\ln(1/\delta)\ln(1/p)}}{\epsilon n} + \frac{B^4d\ln(1/\delta)\ln(1/p)}{\epsilon^2 n^2}\right)\\
    &=O\left(\frac{B^4\sqrt{d}\ln(1/\delta)\ln(1/p)}{\epsilon^2 n}\right)
\end{align*}

By triangle inequality, we have 
\begin{align*}
    |\pvar - \var| &\le |\widetilde{E[X^2]} - \widehat{E[X^2]}| + |\widetilde{E[X]}^2 - \widehat{E[X]}^2|\\
    &= O\left(\frac{B^4\sqrt{d}\ln(1/\delta)\ln(1/p)}{\epsilon^2 n}\right)
\end{align*}
and
$$\tau_1 = O\left(\frac{B^4\sqrt{d}\ln(1/\delta)\ln(1/p)}{\epsilon^2 n\var}\right)$$

For $\widetilde{E[XY]}$ where $X \in R_1$ and $Y \in R_2$, the privatized $n$ in the unbiased estimation is computed as $$d\sqrt{\frac{\widetilde{s[c]}}{d}} = n + o(n) = O(n)$$

Thus, the privatized and non-privatized estimation of $E[XY]$ can be computed as
\begin{align*}
    \widehat{E[XY]} &= \frac{1-n}{1-d} \frac{\sum_{i \in J} \left(\sum_{t \in R_1.i}x\right)\left(\sum_{t \in R_2.i}y\right)}{n^2/d} \\
    & \quad + \frac{n-d}{1-d} \cdot \frac{\sum_{i \in J} \left(\sum_{t \in R_1.i}x\right)}{n^2/d}\cdot\frac{\sum_{i \in J} \left(\sum_{t \in R_2.i}y\right)}{n^2/d}\\
    \widetilde{E[XY]} &= \frac{1-(n+o(n))}{1-d} \frac{\sum_{i \in J} \left(\left(\sum_{t \in R_1.i}x\right) + e_{i,2}\right)\left(\left(\sum_{t \in R_2.i}y\right) + e_{i,3}\right)}{\sum_{j \in J} (n/d + o(n/d))(n/d + o(n/d))} \\
    & + \frac{n+o(n)-d}{1-d} \frac{\sum_{i \in J} \left(\left(\sum_{t \in R_1.i}x\right) + e_{i,2}\right) \sum_{i \in J}\left(\left(\sum_{t \in R_2.i}y\right) + e_{i,3}\right)}{\left(\sum_{j \in J} (n/d + o(n/d))(n/d + o(n/d))\right)^2}\\
    &=\frac{d(1-(n+o(n)))}{n^2(1-d)}\frac{\sum_{i \in J} \left(\left(\sum_{t \in R_1.i}x\right) + e_{i,2}\right)\left(\left(\sum_{t \in R_2.i}y\right) + e_{i,3}\right)}{1 + o(1)}\\
    & + \frac{d^2(n+o(n)-d)}{n^4(1-d)} \frac{\sum_{i \in J} \left(\left(\sum_{t \in R_1.i}x\right) + e_{i,2}\right) \sum_{i \in J}\left(\left(\sum_{t \in R_2.i}y\right) + e_{i,3}\right)}{\left(1 + o(1)\right)^2}
\end{align*}

Based on the same flow of logic as $|\widetilde{E[X^2]} - \widehat{E[X^2]}|$, we would like to bound $\sum_{i \in J} \left(\left(\sum_{t \in R_1.i}x\right)e_{i,3} + \left(\sum_{t \in R_2.i}y\right)e_{i,2} + e_{i,2}e_{i,3}\right)$. Note that 
\begin{align*}
    \sum_{i \in J}\left(\sum_{t \in R_2.i}y\right)e_{i, 2} &\le \frac{nB}{d}\sum_{i \in J}e_{i, 2} = O\left(\frac{ nB^3\sqrt{\ln(1/\delta)\ln(1/p)}}{\sqrt{d}\epsilon}\right)\\
    \sum_{i \in J}e_{i,2}e_{i,3} &\le \sigma\sqrt{\ln(8/p)}\sum_{i \in J}e_{i,2} = O\left(\frac{B^4\sqrt{d}\ln(1/p)\ln(1/\delta)}{\epsilon^2}\right)
\end{align*}

Hence the first two terms are bounded by $O\left(\frac{B^4\ln(1/\delta)\ln(1/p)}{\sqrt{d}\epsilon^2}\right)$. For the last term, we may also bound as
\begin{align*}
    \left(\sum_{i \in J}\sum_{t \in R_1.i}x\right)\sum_{i \in J}e_{i,3} &= O\left(\frac{nB^3\sqrt{d\ln(1/\delta)\ln(1/p)}}{\epsilon}\right)\\
    \sum_{i \in J}e_{i, 3}\sum_{i \in J}e_{i,2} &= O\left(\frac{B^4d\ln(1/\delta)\ln(1/p)}{\epsilon^2}\right)
\end{align*}

So the last term is $O\left(\frac{nB^3d\sqrt{d\ln(1/\delta)\ln(1/p)}}{\epsilon n^3} + \frac{B^4d^2\ln(1/\delta)\ln(1/p)}{\epsilon^2 n^3}\right)$, which can be combined as $O\left(\frac{B^4d^2\ln(1/\delta)\ln(1/p)}{\epsilon^2n^3}\right)$. Therefore $$|\widehat{E[XY]} - \widetilde{E[XY]}| = O\left(\frac{B^4\ln(1/\delta)\ln(1/p)}{\sqrt{d}\epsilon^2}\right)$$

Based on the similar analysis as $|\widetilde{E[X]}^2 - \widehat{E[X]}^2|$, we have $|\widetilde{E[X]}\widetilde{E[Y]} - \widehat{E[X]}\widehat{E[Y]}| =  O\left(\frac{B^4\sqrt{d}\ln(1/\delta)\ln(1/p)}{\epsilon^2 n}\right)$. This yields $$|\cov - \pcov | = O\left(\frac{B^4\ln(1/\delta)\ln(1/p)}{\sqrt{d}\epsilon^2}\right)$$

With an extra assumption that $X$ and $Y$ are 0-centered and each tuple within $R_1$ and $R_2$ is independent and the join key is uncorrelated with $X$ and $Y$. By the Chernoff-Hoeffding's inequality, with probability at least $1-p/4$, we have
\begin{align*}
    \left|\sum_{t \in R_1.i} x\right|, \left|\sum_{t \in R_1.i} y\right| &\le B\sqrt{2\ln(16d/p)n/d} \quad \forall i \in J
\end{align*}

This yields
\begin{align*}
    \sum_{i \in J}\left(\sum_{t \in R_2.i}y\right)e_{i, 2} &\le B\sqrt{2\ln(16d/p)n/d}\sum_{i \in J}e_{i, 2}\\
    &= O\left(\frac{ B^3\sqrt{n\ln(d/p)\ln(1/\delta)\ln(1/p)}}{\epsilon}\right)
\end{align*}
Giving a bound that scale with the size of the relation $$|\widehat{E[XY]} - \widetilde{E[XY]}| = O\left(\frac{B^4\ln(1/p)\ln(1/\delta)\sqrt{d\ln(d/p)}}{\epsilon^2\sqrt{n}}\right)$$

Putting everything together, with probability at least $1-p$, we have 
\begin{align*}
    \tau_2 &= O\left(\frac{B^4\ln(1/p)\ln(1/\delta)\sqrt{d\ln(d/p)}}{\epsilon^2\sqrt{n}\var}\right)
\end{align*}

\stitle{Extension to multi-features.}
The extension of our analysis to multi-dimensional features involves two modifications. Firstly, the bounds $B_1$ and $B_2$ are determined by matrix norm bounds through random matrix theory~\cite{vershynin2018high} instead of the absolute value of single random variable .
Secondly, the bound of the inverse of $\pvar$ is required, where $\pvar$ was scalar but now is a matrix; the inverse of $\pvar$ may become unboundedly large if its minimum eigenvalue is close to 0. To address this, Wang~\cite{wang2018revisiting} makes an additional assumption that the noises to $\var$ has a minimum eigenvalue $\lambda_{min}$ of $o(|\var|)$.

%% file: sections/noise_allocation.tex
\section{Allocation of noises}
\label{app:noiseopt}

We analyze the implication of dynamic allocation of privacy budget for moments on linear regression confidence bound \cref{app:linear}. For union, it is possible to impute noise directly to $(\var)_i, (\cov)_i$, empirical variance, and covariance for each dataset $R_i$. Each of $(\var)_i, (\cov)_i$ has sensitivity $\Delta' = O(B^2/n)$. Thus, let $\sigma'=\sqrt{2\ln(1.25/\delta)}\Delta'/\epsilon$ and $$\pvar \sim \frac{\sum_{i}^k \left( (\var)_i + \calN(0, \sigma'^2)\right)}{k}$$

Applying gaussian tail bound and the independency assumption yields $|\pvar - \var| = O(B^2\sqrt{\ln(1/\delta)\ln(1/p)}/\epsilon\sqrt{k}n)$, and $|\pcov - \cov| = O(B^2\sqrt{\ln(1/\delta)\ln(1/p)}/\epsilon\sqrt{k}n)$. This reduces the bound on $\tau_1$ and $\tau_2$ by a factor of $O(B^2\sqrt{\ln(1/\delta)\ln(1/p)}/\epsilon)$.
Based on \cref{sensitivity}, consider the query $q_{j, i}: \calD^n \longrightarrow \mathcal{S}^i$  where $\mathcal{S}^i = \{v \ | \ v \in \mathbb{R}^i\}$. $q_{j, i}$ returns a vector $s_i \in \mathcal{S}^i$ containing the sum of the $i$-order monomials across each join key. $$\Delta_{q_{j,i}} = \sqrt{1_{2\mathbb{Z} + 1}(i)4B^{2i} + 1_{2\mathbb{Z}}(i)2B^{2i}}$$

For linear regression, it is feasible to decomposite $q_j$ into 3 sequential queries, $q_{j, 0}$, $q_{j, 1}$ and $q_{j, 2}$, each with privacy budget $(\epsilon/3, \delta/3)$. Inheriting notations from \cref{app:linear},  
$\Delta = \Delta_{q_j}$, $O(B^2\Delta_{q_{j, 0}}) = O(B\Delta_{q_{j, 1}}) = 
O(\Delta_{q_{j, 2}}) = O(B^2)$, note that although there is less privacy budget on releasing the count of tuples within each join key, the sensitivity is also reduced by a magnitude of $B^2$, i.e. from $\Delta$ to $\Delta_{q_{j 0}}$. Hence, it is reasonable to assume that the noise on this number is small, and bounded by $o(n/d)$. The main implication is that $e_{i, 2}, e_{i, 3} = \calN(0, 2\ln(1.25/(\delta/3))\Delta^2_{q_{j, 1}}/(\epsilon / 3)^2)$, and $e_{i, 1} = \calN(0, 2\ln(1.25/(\delta/3))\Delta^2_{q_{j, 2}}/(\epsilon / 3)^2)$, and no more change to the analysis is required. Following the computations in \cref{app:linear}, we have $$\tau_1 = O(\frac{B^2\sqrt{d}\ln(1/\delta)\ln(1/p)}{\epsilon^2 n \var}), \tau_2 = O(\frac{B^2\ln(1/\delta)\ln(1/p)}{\sqrt{d}\widehat{\sigma_x^2}})$$

%% file: main.bbl

\begin{thebibliography}{73}


\ifx \showCODEN    \undefined \def \showCODEN     #1{\unskip}     \fi
\ifx \showDOI      \undefined \def \showDOI       #1{#1}\fi
\ifx \showISBNx    \undefined \def \showISBNx     #1{\unskip}     \fi
\ifx \showISBNxiii \undefined \def \showISBNxiii  #1{\unskip}     \fi
\ifx \showISSN     \undefined \def \showISSN      #1{\unskip}     \fi
\ifx \showLCCN     \undefined \def \showLCCN      #1{\unskip}     \fi
\ifx \shownote     \undefined \def \shownote      #1{#1}          \fi
\ifx \showarticletitle \undefined \def \showarticletitle #1{#1}   \fi
\ifx \showURL      \undefined \def \showURL       {\relax}        \fi
\providecommand\bibfield[2]{#2}
\providecommand\bibinfo[2]{#2}
\providecommand\natexlab[1]{#1}
\providecommand\showeprint[2][]{arXiv:#2}

\bibitem[ELA({[n.\,d.]})]%
        {ELA}
 \bibinfo{year}{[n.\,d.]}\natexlab{}.
\newblock \bibinfo{title}{2013 - 2018 School ELA REsults}.
\newblock
  \bibinfo{howpublished}{\url{https://data.cityofnewyork.us/Education/2013-2018-School-ELA-REsults/qkpp-pbi8}}.
\newblock


\bibitem[Mat({[n.\,d.]})]%
        {Math}
 \bibinfo{year}{[n.\,d.]}\natexlab{}.
\newblock \bibinfo{title}{2013 -2018 School Math Results}.
\newblock
  \bibinfo{howpublished}{\url{https://data.cityofnewyork.us/Education/2013-2018-School-Math-Results/m27t-ht3h}}.
\newblock


\bibitem[gen({[n.\,d.]})]%
        {gender}
 \bibinfo{year}{[n.\,d.]}\natexlab{}.
\newblock \bibinfo{title}{2013-16 School ELA Data Files By Grade - Gender}.
\newblock
  \bibinfo{howpublished}{\url{https://data.cityofnewyork.us/Education/2013-16-School-ELA-Data-Files-By-Grade-Gender/436j-ja87}}.
\newblock


\bibitem[Reg({[n.\,d.]})]%
        {Regent}
 \bibinfo{year}{[n.\,d.]}\natexlab{}.
\newblock \bibinfo{title}{2014-15 To 2016-17 School- Level NYC Regents Report
  For All Variables}.
\newblock
  \bibinfo{howpublished}{\url{https://data.cityofnewyork.us/Education/2014-15-To-2016-17-School-Level-NYC-Regents-Report/csps-2ne9/}}.
\newblock


\bibitem[Gra({[n.\,d.]})]%
        {Grad}
 \bibinfo{year}{[n.\,d.]}\natexlab{}.
\newblock \bibinfo{title}{2016-2017 Graduation Outcomes School}.
\newblock
  \bibinfo{howpublished}{\url{https://data.cityofnewyork.us/Education/2016-2017-Graduation-Outcomes-School/nb39-jx2v}}.
\newblock


\bibitem[CCP({[n.\,d.]})]%
        {CCPA}
 \bibinfo{year}{[n.\,d.]}\natexlab{}.
\newblock \bibinfo{title}{California Consumer Privacy Act}.
\newblock \bibinfo{howpublished}{\url{https://oag.ca.gov/privacy/ccpa}}.
\newblock


\bibitem[FER({[n.\,d.]})]%
        {FERPA}
 \bibinfo{year}{[n.\,d.]}\natexlab{}.
\newblock \bibinfo{title}{The Family Educational Rights and Privacy Act
  (FERPA)}.
\newblock \bibinfo{howpublished}{\url{https://studentprivacy.ed.gov/}}.
\newblock


\bibitem[HIP({[n.\,d.]})]%
        {HIPAA}
 \bibinfo{year}{[n.\,d.]}\natexlab{}.
\newblock \bibinfo{title}{Health Insurance Portability and Accountability Act
  of 1996 (HIPAA)}.
\newblock
  \bibinfo{howpublished}{\url{https://www.cdc.gov/phlp/publications/topic/hipaa.html}}.
\newblock


\bibitem[EUd(2018)]%
        {EUdataregulations2018}
 \bibinfo{year}{2018}\natexlab{}.
\newblock \bibinfo{title}{2018 reform of EU data protection rules}.
\newblock
  \bibinfo{howpublished}{\url{https://ec.europa.eu/commission/sites/beta-political/files/data-protection-factsheet-changes_en.pdf}}.
\newblock


\bibitem[cms(2022)]%
        {cms}
 \bibinfo{year}{2022}\natexlab{}.
\newblock \bibinfo{title}{CMS Data}.
\newblock \bibinfo{howpublished}{\url{https://data.cms.gov/}}.
\newblock


\bibitem[nyc(2022)]%
        {nycopen}
 \bibinfo{year}{2022}\natexlab{}.
\newblock \bibinfo{title}{NYC Open Data}.
\newblock \bibinfo{howpublished}{\url{https://opendata.cityofnewyork.us/}}.
\newblock


\bibitem[tec(2023)]%
        {tech}
 \bibinfo{year}{2023}\natexlab{}.
\newblock \bibinfo{title}{(Technical Report) Saibot: A Differentially Private
  Data Search Platform}.
\newblock
  \bibinfo{howpublished}{\url{https://anonymous.4open.science/r/Saibot-B387/tech/saibot_tech.pdf}}.
\newblock


\bibitem[Abo~Khamis et~al\mbox{.}(2016)]%
        {abo2016faq}
\bibfield{author}{\bibinfo{person}{Mahmoud Abo~Khamis}, \bibinfo{person}{Hung~Q
  Ngo}, {and} \bibinfo{person}{Atri Rudra}.} \bibinfo{year}{2016}\natexlab{}.
\newblock \showarticletitle{FAQ: questions asked frequently}. In
  \bibinfo{booktitle}{\emph{Proceedings of the 35th ACM SIGMOD-SIGACT-SIGAI
  Symposium on Principles of Database Systems}}. \bibinfo{pages}{13--28}.
\newblock


\bibitem[Alabi and Vadhan(2022)]%
        {NEURIPS2022_5bc3356e}
\bibfield{author}{\bibinfo{person}{Daniel Alabi} {and} \bibinfo{person}{Salil
  Vadhan}.} \bibinfo{year}{2022}\natexlab{}.
\newblock \showarticletitle{Hypothesis Testing for Differentially Private
  Linear Regression}. In \bibinfo{booktitle}{\emph{Advances in Neural
  Information Processing Systems}}, Vol.~\bibinfo{volume}{35}.
  \bibinfo{pages}{14196--14209}.
\newblock
\urldef\tempurl%
\url{https://proceedings.neurips.cc/paper_files/paper/2022/file/5bc3356e0fa1753fff7e8d6628e71b22-Paper-Conference.pdf}
\showURL{%
\tempurl}


\bibitem[Alabi(2022)]%
        {Alabi22}
\bibfield{author}{\bibinfo{person}{Daniel~Gbenga Alabi}.}
  \bibinfo{year}{2022}\natexlab{}.
\newblock \emph{\bibinfo{title}{The Algorithmic Foundations of Private
  Computational Social Science}}.
\newblock \bibinfo{thesistype}{Ph.\,D. Dissertation}. \bibinfo{school}{Harvard
  University}.
\newblock


\bibitem[Blum et~al\mbox{.}(2013)]%
        {blum2013learning}
\bibfield{author}{\bibinfo{person}{Avrim Blum}, \bibinfo{person}{Katrina
  Ligett}, {and} \bibinfo{person}{Aaron Roth}.}
  \bibinfo{year}{2013}\natexlab{}.
\newblock \showarticletitle{A learning theory approach to noninteractive
  database privacy}.
\newblock \bibinfo{journal}{\emph{Journal of the ACM (JACM)}}
  \bibinfo{volume}{60}, \bibinfo{number}{2} (\bibinfo{year}{2013}),
  \bibinfo{pages}{1--25}.
\newblock


\bibitem[Castelo et~al\mbox{.}(2021)]%
        {castelo2021auctus}
\bibfield{author}{\bibinfo{person}{Sonia Castelo}, \bibinfo{person}{R{\'e}mi
  Rampin}, \bibinfo{person}{A{\'e}cio Santos}, \bibinfo{person}{Aline Bessa},
  \bibinfo{person}{Fernando Chirigati}, {and} \bibinfo{person}{Juliana
  Freire}.} \bibinfo{year}{2021}\natexlab{}.
\newblock \showarticletitle{Auctus: a dataset search engine for data discovery
  and augmentation}.
\newblock \bibinfo{journal}{\emph{Proceedings of the VLDB Endowment}}
  \bibinfo{volume}{14}, \bibinfo{number}{12} (\bibinfo{year}{2021}),
  \bibinfo{pages}{2791--2794}.
\newblock


\bibitem[Cerda and Varoquaux(2020)]%
        {cerda2020encoding}
\bibfield{author}{\bibinfo{person}{Patricio Cerda} {and}
  \bibinfo{person}{Ga{\"e}l Varoquaux}.} \bibinfo{year}{2020}\natexlab{}.
\newblock \showarticletitle{Encoding high-cardinality string categorical
  variables}.
\newblock \bibinfo{journal}{\emph{IEEE Transactions on Knowledge and Data
  Engineering}} \bibinfo{volume}{34}, \bibinfo{number}{3}
  (\bibinfo{year}{2020}), \bibinfo{pages}{1164--1176}.
\newblock


\bibitem[Chen et~al\mbox{.}(2017)]%
        {chen2017semi}
\bibfield{author}{\bibinfo{person}{Xiaojun Chen}, \bibinfo{person}{Guowen
  Yuan}, \bibinfo{person}{Feiping Nie}, {and} \bibinfo{person}{Joshua~Zhexue
  Huang}.} \bibinfo{year}{2017}\natexlab{}.
\newblock \showarticletitle{Semi-supervised Feature Selection via Rescaled
  Linear Regression.}. In \bibinfo{booktitle}{\emph{IJCAI}},
  Vol.~\bibinfo{volume}{2017}. \bibinfo{pages}{1525--1531}.
\newblock


\bibitem[Chepurko et~al\mbox{.}(2020)]%
        {chepurko2020arda}
\bibfield{author}{\bibinfo{person}{Nadiia Chepurko}, \bibinfo{person}{Ryan
  Marcus}, \bibinfo{person}{Emanuel Zgraggen}, \bibinfo{person}{Raul~Castro
  Fernandez}, \bibinfo{person}{Tim Kraska}, {and} \bibinfo{person}{David
  Karger}.} \bibinfo{year}{2020}\natexlab{}.
\newblock \showarticletitle{ARDA: automatic relational data augmentation for
  machine learning}.
\newblock \bibinfo{journal}{\emph{arXiv preprint arXiv:2003.09758}}
  (\bibinfo{year}{2020}).
\newblock


\bibitem[Ding et~al\mbox{.}(2017)]%
        {ding2017collecting}
\bibfield{author}{\bibinfo{person}{Bolin Ding}, \bibinfo{person}{Janardhan
  Kulkarni}, {and} \bibinfo{person}{Sergey Yekhanin}.}
  \bibinfo{year}{2017}\natexlab{}.
\newblock \showarticletitle{Collecting telemetry data privately}.
\newblock \bibinfo{journal}{\emph{Advances in Neural Information Processing
  Systems}}  \bibinfo{volume}{30} (\bibinfo{year}{2017}).
\newblock


\bibitem[Dwork et~al\mbox{.}(2006a)]%
        {dwork2006our}
\bibfield{author}{\bibinfo{person}{Cynthia Dwork}, \bibinfo{person}{Krishnaram
  Kenthapadi}, \bibinfo{person}{Frank McSherry}, \bibinfo{person}{Ilya
  Mironov}, {and} \bibinfo{person}{Moni Naor}.}
  \bibinfo{year}{2006}\natexlab{a}.
\newblock \showarticletitle{Our data, ourselves: Privacy via distributed noise
  generation}. In \bibinfo{booktitle}{\emph{Annual international conference on
  the theory and applications of cryptographic techniques}}. Springer,
  \bibinfo{pages}{486--503}.
\newblock


\bibitem[Dwork et~al\mbox{.}(2006b)]%
        {dwork2006calibrating}
\bibfield{author}{\bibinfo{person}{Cynthia Dwork}, \bibinfo{person}{Frank
  McSherry}, \bibinfo{person}{Kobbi Nissim}, {and} \bibinfo{person}{Adam
  Smith}.} \bibinfo{year}{2006}\natexlab{b}.
\newblock \showarticletitle{Calibrating noise to sensitivity in private data
  analysis}. In \bibinfo{booktitle}{\emph{Theory of Cryptography: Third Theory
  of Cryptography Conference, TCC 2006, New York, NY, USA, March 4-7, 2006.
  Proceedings 3}}. Springer, \bibinfo{pages}{265--284}.
\newblock


\bibitem[Dwork et~al\mbox{.}(2010)]%
        {dwork2010boosting}
\bibfield{author}{\bibinfo{person}{Cynthia Dwork}, \bibinfo{person}{Guy~N
  Rothblum}, {and} \bibinfo{person}{Salil Vadhan}.}
  \bibinfo{year}{2010}\natexlab{}.
\newblock \showarticletitle{Boosting and differential privacy}. In
  \bibinfo{booktitle}{\emph{2010 IEEE 51st Annual Symposium on Foundations of
  Computer Science}}. IEEE, \bibinfo{pages}{51--60}.
\newblock


\bibitem[Dwork et~al\mbox{.}(2017)]%
        {dwork2017exposed}
\bibfield{author}{\bibinfo{person}{Cynthia Dwork}, \bibinfo{person}{Adam
  Smith}, \bibinfo{person}{Thomas Steinke}, {and} \bibinfo{person}{Jonathan
  Ullman}.} \bibinfo{year}{2017}\natexlab{}.
\newblock \showarticletitle{Exposed! a survey of attacks on private data}.
\newblock \bibinfo{journal}{\emph{Annual Review of Statistics and Its
  Application}}  \bibinfo{volume}{4} (\bibinfo{year}{2017}),
  \bibinfo{pages}{61--84}.
\newblock


\bibitem[Dwork et~al\mbox{.}(2014)]%
        {dwork2014analyze}
\bibfield{author}{\bibinfo{person}{Cynthia Dwork}, \bibinfo{person}{Kunal
  Talwar}, \bibinfo{person}{Abhradeep Thakurta}, {and} \bibinfo{person}{Li
  Zhang}.} \bibinfo{year}{2014}\natexlab{}.
\newblock \showarticletitle{Analyze gauss: optimal bounds for
  privacy-preserving principal component analysis}. In
  \bibinfo{booktitle}{\emph{Proceedings of the forty-sixth annual ACM symposium
  on Theory of computing}}. \bibinfo{pages}{11--20}.
\newblock


\bibitem[Erlingsson et~al\mbox{.}(2019)]%
        {erlingsson2019amplification}
\bibfield{author}{\bibinfo{person}{{\'U}lfar Erlingsson},
  \bibinfo{person}{Vitaly Feldman}, \bibinfo{person}{Ilya Mironov},
  \bibinfo{person}{Ananth Raghunathan}, \bibinfo{person}{Kunal Talwar}, {and}
  \bibinfo{person}{Abhradeep Thakurta}.} \bibinfo{year}{2019}\natexlab{}.
\newblock \showarticletitle{Amplification by shuffling: From local to central
  differential privacy via anonymity}. In \bibinfo{booktitle}{\emph{Proceedings
  of the Thirtieth Annual ACM-SIAM Symposium on Discrete Algorithms}}. SIAM,
  \bibinfo{pages}{2468--2479}.
\newblock


\bibitem[Erlingsson et~al\mbox{.}(2014)]%
        {erlingsson2014rappor}
\bibfield{author}{\bibinfo{person}{{\'U}lfar Erlingsson},
  \bibinfo{person}{Vasyl Pihur}, {and} \bibinfo{person}{Aleksandra Korolova}.}
  \bibinfo{year}{2014}\natexlab{}.
\newblock \showarticletitle{Rappor: Randomized aggregatable privacy-preserving
  ordinal response}. In \bibinfo{booktitle}{\emph{Proceedings of the 2014 ACM
  SIGSAC conference on computer and communications security}}.
  \bibinfo{pages}{1054--1067}.
\newblock


\bibitem[Feldman et~al\mbox{.}(2022)]%
        {feldman2022hiding}
\bibfield{author}{\bibinfo{person}{Vitaly Feldman}, \bibinfo{person}{Audra
  McMillan}, {and} \bibinfo{person}{Kunal Talwar}.}
  \bibinfo{year}{2022}\natexlab{}.
\newblock \showarticletitle{Hiding among the clones: A simple and nearly
  optimal analysis of privacy amplification by shuffling}. In
  \bibinfo{booktitle}{\emph{2021 IEEE 62nd Annual Symposium on Foundations of
  Computer Science (FOCS)}}. IEEE, \bibinfo{pages}{954--964}.
\newblock


\bibitem[Fernandez et~al\mbox{.}(2018)]%
        {fernandez2018aurum}
\bibfield{author}{\bibinfo{person}{Raul~Castro Fernandez},
  \bibinfo{person}{Ziawasch Abedjan}, \bibinfo{person}{Famien Koko},
  \bibinfo{person}{Gina Yuan}, \bibinfo{person}{Samuel Madden}, {and}
  \bibinfo{person}{Michael Stonebraker}.} \bibinfo{year}{2018}\natexlab{}.
\newblock \showarticletitle{Aurum: A data discovery system}. In
  \bibinfo{booktitle}{\emph{2018 IEEE 34th International Conference on Data
  Engineering (ICDE)}}. IEEE, \bibinfo{pages}{1001--1012}.
\newblock


\bibitem[Frank et~al\mbox{.}(2007)]%
        {frank2007method}
\bibfield{author}{\bibinfo{person}{Richard Frank}, \bibinfo{person}{Flavia
  Moser}, {and} \bibinfo{person}{Martin Ester}.}
  \bibinfo{year}{2007}\natexlab{}.
\newblock \showarticletitle{A method for multi-relational classification using
  single and multi-feature aggregation functions}. In
  \bibinfo{booktitle}{\emph{Knowledge Discovery in Databases: PKDD 2007: 11th
  European Conference on Principles and Practice of Knowledge Discovery in
  Databases, Warsaw, Poland, September 17-21, 2007. Proceedings 11}}. Springer,
  \bibinfo{pages}{430--437}.
\newblock


\bibitem[Green et~al\mbox{.}(2007)]%
        {green2007provenance}
\bibfield{author}{\bibinfo{person}{Todd~J Green}, \bibinfo{person}{Grigoris
  Karvounarakis}, {and} \bibinfo{person}{Val Tannen}.}
  \bibinfo{year}{2007}\natexlab{}.
\newblock \showarticletitle{Provenance semirings}. In
  \bibinfo{booktitle}{\emph{Proceedings of the twenty-sixth ACM
  SIGMOD-SIGACT-SIGART symposium on Principles of database systems}}.
  \bibinfo{pages}{31--40}.
\newblock


\bibitem[Guo and Viktor(2008)]%
        {guo2008multirelational}
\bibfield{author}{\bibinfo{person}{Hongyu Guo} {and} \bibinfo{person}{Herna~L
  Viktor}.} \bibinfo{year}{2008}\natexlab{}.
\newblock \showarticletitle{Multirelational classification: a multiple view
  approach}.
\newblock \bibinfo{journal}{\emph{Knowledge and Information Systems}}
  \bibinfo{volume}{17} (\bibinfo{year}{2008}), \bibinfo{pages}{287--312}.
\newblock


\bibitem[Hardt and Rothblum(2010)]%
        {hardt2010multiplicative}
\bibfield{author}{\bibinfo{person}{Moritz Hardt} {and} \bibinfo{person}{Guy~N
  Rothblum}.} \bibinfo{year}{2010}\natexlab{}.
\newblock \showarticletitle{A multiplicative weights mechanism for
  privacy-preserving data analysis}. In \bibinfo{booktitle}{\emph{2010 IEEE
  51st annual symposium on foundations of computer science}}. IEEE,
  \bibinfo{pages}{61--70}.
\newblock


\bibitem[Hardy et~al\mbox{.}(2017)]%
        {hardy2017private}
\bibfield{author}{\bibinfo{person}{Stephen Hardy}, \bibinfo{person}{Wilko
  Henecka}, \bibinfo{person}{Hamish Ivey-Law}, \bibinfo{person}{Richard Nock},
  \bibinfo{person}{Giorgio Patrini}, \bibinfo{person}{Guillaume Smith}, {and}
  \bibinfo{person}{Brian Thorne}.} \bibinfo{year}{2017}\natexlab{}.
\newblock \showarticletitle{Private federated learning on vertically
  partitioned data via entity resolution and additively homomorphic
  encryption}.
\newblock \bibinfo{journal}{\emph{arXiv preprint arXiv:1711.10677}}
  (\bibinfo{year}{2017}).
\newblock


\bibitem[Huang et~al\mbox{.}(2023a)]%
        {weighing}
\bibfield{author}{\bibinfo{person}{Zezhou Huang}, \bibinfo{person}{Pavan~Kalyan
  Damalapati}, {and} \bibinfo{person}{Eugene Wu}.}
  \bibinfo{year}{2023}\natexlab{a}.
\newblock \showarticletitle{Aggregation Consistency Errors in Semantic Layers
  and How to Avoid Them}. In \bibinfo{booktitle}{\emph{Proceedings of the
  Workshop on Human-In-the-Loop Data Analytics}}.
\newblock


\bibitem[Huang et~al\mbox{.}(2023b)]%
        {joinboost}
\bibfield{author}{\bibinfo{person}{Zezhou Huang}, \bibinfo{person}{Rathijit
  Sen}, \bibinfo{person}{Jiaxiang Liu}, {and} \bibinfo{person}{Eugene Wu}.}
  \bibinfo{year}{2023}\natexlab{b}.
\newblock \showarticletitle{JoinBoost: Grow Trees Over Normalized Data Using
  Only SQL}.
\newblock \bibinfo{journal}{\emph{VLDB}}.
\newblock


\bibitem[Huang et~al\mbox{.}(2023c)]%
        {kitana}
\bibfield{author}{\bibinfo{person}{Zezhou Huang}, \bibinfo{person}{Pranav
  Subramaniam}, \bibinfo{person}{Raul~Castro Fernandez}, {and}
  \bibinfo{person}{Eugene Wu}.} \bibinfo{year}{2023}\natexlab{c}.
\newblock \bibinfo{title}{Kitana: Efficient Data Augmentation Search for
  AutoML}.
\newblock
\newblock
\showeprint[arxiv]{2305.10419}~[cs.DB]


\bibitem[Huggins et~al\mbox{.}(2017)]%
        {huggins2017pass}
\bibfield{author}{\bibinfo{person}{Jonathan Huggins}, \bibinfo{person}{Ryan~P
  Adams}, {and} \bibinfo{person}{Tamara Broderick}.}
  \bibinfo{year}{2017}\natexlab{}.
\newblock \showarticletitle{PASS-GLM: polynomial approximate sufficient
  statistics for scalable Bayesian GLM inference}.
\newblock \bibinfo{journal}{\emph{Advances in Neural Information Processing
  Systems}}  \bibinfo{volume}{30} (\bibinfo{year}{2017}).
\newblock


\bibitem[Johnson et~al\mbox{.}(2018)]%
        {johnson2018towards}
\bibfield{author}{\bibinfo{person}{Noah Johnson}, \bibinfo{person}{Joseph~P
  Near}, {and} \bibinfo{person}{Dawn Song}.} \bibinfo{year}{2018}\natexlab{}.
\newblock \showarticletitle{Towards practical differential privacy for SQL
  queries}.
\newblock \bibinfo{journal}{\emph{Proceedings of the VLDB Endowment}}
  \bibinfo{volume}{11}, \bibinfo{number}{5} (\bibinfo{year}{2018}),
  \bibinfo{pages}{526--539}.
\newblock


\bibitem[Kairouz et~al\mbox{.}(2016)]%
        {kairouz2016discrete}
\bibfield{author}{\bibinfo{person}{Peter Kairouz}, \bibinfo{person}{Keith
  Bonawitz}, {and} \bibinfo{person}{Daniel Ramage}.}
  \bibinfo{year}{2016}\natexlab{}.
\newblock \showarticletitle{Discrete distribution estimation under local
  privacy}. In \bibinfo{booktitle}{\emph{International Conference on Machine
  Learning}}. PMLR, \bibinfo{pages}{2436--2444}.
\newblock


\bibitem[Khamis et~al\mbox{.}(2020)]%
        {khamis2020functional}
\bibfield{author}{\bibinfo{person}{Mahmoud~Abo Khamis}, \bibinfo{person}{Ryan~R
  Curtin}, \bibinfo{person}{Benjamin Moseley}, \bibinfo{person}{Hung~Q Ngo},
  \bibinfo{person}{XuanLong Nguyen}, \bibinfo{person}{Dan Olteanu}, {and}
  \bibinfo{person}{Maximilian Schleich}.} \bibinfo{year}{2020}\natexlab{}.
\newblock \showarticletitle{Functional Aggregate Queries with Additive
  Inequalities}.
\newblock \bibinfo{journal}{\emph{ACM Transactions on Database Systems (TODS)}}
  \bibinfo{volume}{45}, \bibinfo{number}{4} (\bibinfo{year}{2020}),
  \bibinfo{pages}{1--41}.
\newblock


\bibitem[Khamis et~al\mbox{.}(2018)]%
        {khamis2018ac}
\bibfield{author}{\bibinfo{person}{Mahmoud~Abo Khamis}, \bibinfo{person}{Hung~Q
  Ngo}, \bibinfo{person}{XuanLong Nguyen}, \bibinfo{person}{Dan Olteanu}, {and}
  \bibinfo{person}{Maximilian Schleich}.} \bibinfo{year}{2018}\natexlab{}.
\newblock \showarticletitle{AC/DC: in-database learning thunderstruck}. In
  \bibinfo{booktitle}{\emph{Proceedings of the second workshop on data
  management for end-to-end machine learning}}. \bibinfo{pages}{1--10}.
\newblock


\bibitem[Kotsogiannis et~al\mbox{.}(2019)]%
        {kotsogiannis2019privatesql}
\bibfield{author}{\bibinfo{person}{Ios Kotsogiannis}, \bibinfo{person}{Yuchao
  Tao}, \bibinfo{person}{Xi He}, \bibinfo{person}{Maryam Fanaeepour},
  \bibinfo{person}{Ashwin Machanavajjhala}, \bibinfo{person}{Michael Hay},
  {and} \bibinfo{person}{Gerome Miklau}.} \bibinfo{year}{2019}\natexlab{}.
\newblock \showarticletitle{Privatesql: a differentially private sql query
  engine}.
\newblock \bibinfo{journal}{\emph{Proceedings of the VLDB Endowment}}
  \bibinfo{volume}{12}, \bibinfo{number}{11} (\bibinfo{year}{2019}),
  \bibinfo{pages}{1371--1384}.
\newblock


\bibitem[Kulkarni et~al\mbox{.}(2021)]%
        {kulkarni2021differentially}
\bibfield{author}{\bibinfo{person}{Tejas Kulkarni}, \bibinfo{person}{Joonas
  J{\"a}lk{\"o}}, \bibinfo{person}{Antti Koskela}, \bibinfo{person}{Samuel
  Kaski}, {and} \bibinfo{person}{Antti Honkela}.}
  \bibinfo{year}{2021}\natexlab{}.
\newblock \showarticletitle{Differentially private bayesian inference for
  generalized linear models}. In \bibinfo{booktitle}{\emph{International
  Conference on Machine Learning}}. PMLR, \bibinfo{pages}{5838--5849}.
\newblock


\bibitem[Lee and Clifton(2011)]%
        {lee2011much}
\bibfield{author}{\bibinfo{person}{Jaewoo Lee} {and} \bibinfo{person}{Chris
  Clifton}.} \bibinfo{year}{2011}\natexlab{}.
\newblock \showarticletitle{How much is enough? choosing $\varepsilon$ for
  differential privacy}. In \bibinfo{booktitle}{\emph{Information Security:
  14th International Conference, ISC 2011, Xi’an, China, October 26-29, 2011.
  Proceedings 14}}. Springer, \bibinfo{pages}{325--340}.
\newblock


\bibitem[Li et~al\mbox{.}(2021)]%
        {li2021data}
\bibfield{author}{\bibinfo{person}{Yifan Li}, \bibinfo{person}{Xiaohui Yu},
  {and} \bibinfo{person}{Nick Koudas}.} \bibinfo{year}{2021}\natexlab{}.
\newblock \showarticletitle{Data acquisition for improving machine learning
  models}.
\newblock \bibinfo{journal}{\emph{Proceedings of the VLDB Endowment}}
  \bibinfo{volume}{14}, \bibinfo{number}{10} (\bibinfo{year}{2021}),
  \bibinfo{pages}{1832--1844}.
\newblock


\bibitem[Ma{\'c}kiewicz and Ratajczak(1993)]%
        {mackiewicz1993principal}
\bibfield{author}{\bibinfo{person}{Andrzej Ma{\'c}kiewicz} {and}
  \bibinfo{person}{Waldemar Ratajczak}.} \bibinfo{year}{1993}\natexlab{}.
\newblock \showarticletitle{Principal components analysis (PCA)}.
\newblock \bibinfo{journal}{\emph{Computers \& Geosciences}}
  \bibinfo{volume}{19}, \bibinfo{number}{3} (\bibinfo{year}{1993}),
  \bibinfo{pages}{303--342}.
\newblock


\bibitem[Moeyersoms and Martens(2015)]%
        {moeyersoms2015including}
\bibfield{author}{\bibinfo{person}{Julie Moeyersoms} {and}
  \bibinfo{person}{David Martens}.} \bibinfo{year}{2015}\natexlab{}.
\newblock \showarticletitle{Including high-cardinality attributes in predictive
  models: A case study in churn prediction in the energy sector}.
\newblock \bibinfo{journal}{\emph{Decision support systems}}
  \bibinfo{volume}{72} (\bibinfo{year}{2015}), \bibinfo{pages}{72--81}.
\newblock


\bibitem[Nargesian et~al\mbox{.}(2022)]%
        {nargesian2022responsible}
\bibfield{author}{\bibinfo{person}{Fatemeh Nargesian},
  \bibinfo{person}{Abolfazl Asudeh}, {and} \bibinfo{person}{HV Jagadish}.}
  \bibinfo{year}{2022}\natexlab{}.
\newblock \showarticletitle{Responsible Data Integration: Next-generation
  Challenges}. In \bibinfo{booktitle}{\emph{Proceedings of the 2022
  International Conference on Management of Data}}.
  \bibinfo{pages}{2458--2464}.
\newblock


\bibitem[Near et~al\mbox{.}(2021)]%
        {near2021differential}
\bibfield{author}{\bibinfo{person}{Joseph~P Near}, \bibinfo{person}{Xi He},
  {et~al\mbox{.}}} \bibinfo{year}{2021}\natexlab{}.
\newblock \showarticletitle{Differential Privacy for Databases}.
\newblock \bibinfo{journal}{\emph{Foundations and Trends{\textregistered} in
  Databases}} \bibinfo{volume}{11}, \bibinfo{number}{2} (\bibinfo{year}{2021}),
  \bibinfo{pages}{109--225}.
\newblock


\bibitem[North(2019)]%
        {fitbit}
\bibfield{author}{\bibinfo{person}{Ted North}.}
  \bibinfo{year}{2019}\natexlab{}.
\newblock \bibinfo{title}{Google, Fitbit, and the Sale of Our Private Health
  Data}.
\newblock \bibinfo{howpublished}{\url{https://www.fitbit.com/global/us/home}}.
\newblock


\bibitem[Olteanu and Z{\'a}vodn{\`y}(2015)]%
        {olteanu2015size}
\bibfield{author}{\bibinfo{person}{Dan Olteanu} {and} \bibinfo{person}{Jakub
  Z{\'a}vodn{\`y}}.} \bibinfo{year}{2015}\natexlab{}.
\newblock \showarticletitle{Size bounds for factorised representations of query
  results}.
\newblock \bibinfo{journal}{\emph{ACM Transactions on Database Systems (TODS)}}
  \bibinfo{volume}{40}, \bibinfo{number}{1} (\bibinfo{year}{2015}),
  \bibinfo{pages}{1--44}.
\newblock


\bibitem[Pearl(2022)]%
        {pearl2022comment}
\bibfield{author}{\bibinfo{person}{Judea Pearl}.}
  \bibinfo{year}{2022}\natexlab{}.
\newblock \showarticletitle{Comment: understanding Simpson’s paradox}.
\newblock In \bibinfo{booktitle}{\emph{Probabilistic and Causal Inference: The
  Works of Judea Pearl}}. \bibinfo{pages}{399--412}.
\newblock


\bibitem[Qardaji et~al\mbox{.}(2014)]%
        {qardaji2014priview}
\bibfield{author}{\bibinfo{person}{Wahbeh Qardaji}, \bibinfo{person}{Weining
  Yang}, {and} \bibinfo{person}{Ninghui Li}.} \bibinfo{year}{2014}\natexlab{}.
\newblock \showarticletitle{Priview: practical differentially private release
  of marginal contingency tables}. In \bibinfo{booktitle}{\emph{Proceedings of
  the 2014 ACM SIGMOD international conference on Management of data}}.
  \bibinfo{pages}{1435--1446}.
\newblock


\bibitem[Sambasivan et~al\mbox{.}(2021)]%
        {sambasivan2021everyone}
\bibfield{author}{\bibinfo{person}{Nithya Sambasivan}, \bibinfo{person}{Shivani
  Kapania}, \bibinfo{person}{Hannah Highfill}, \bibinfo{person}{Diana Akrong},
  \bibinfo{person}{Praveen Paritosh}, {and} \bibinfo{person}{Lora~M Aroyo}.}
  \bibinfo{year}{2021}\natexlab{}.
\newblock \showarticletitle{“Everyone wants to do the model work, not the
  data work”: Data Cascades in High-Stakes AI}. In
  \bibinfo{booktitle}{\emph{proceedings of the 2021 CHI Conference on Human
  Factors in Computing Systems}}. \bibinfo{pages}{1--15}.
\newblock


\bibitem[Santos et~al\mbox{.}(2022)]%
        {santos2022sketch}
\bibfield{author}{\bibinfo{person}{A{\'e}cio Santos}, \bibinfo{person}{Aline
  Bessa}, \bibinfo{person}{Christopher Musco}, {and} \bibinfo{person}{Juliana
  Freire}.} \bibinfo{year}{2022}\natexlab{}.
\newblock \showarticletitle{A sketch-based index for correlated dataset
  search}. In \bibinfo{booktitle}{\emph{2022 IEEE 38th International Conference
  on Data Engineering (ICDE)}}. IEEE, \bibinfo{pages}{2928--2941}.
\newblock


\bibitem[Schleich et~al\mbox{.}(2019)]%
        {schleich2019layered}
\bibfield{author}{\bibinfo{person}{Maximilian Schleich}, \bibinfo{person}{Dan
  Olteanu}, \bibinfo{person}{Mahmoud Abo~Khamis}, \bibinfo{person}{Hung~Q Ngo},
  {and} \bibinfo{person}{XuanLong Nguyen}.} \bibinfo{year}{2019}\natexlab{}.
\newblock \showarticletitle{A layered aggregate engine for analytics
  workloads}. In \bibinfo{booktitle}{\emph{Proceedings of the 2019
  International Conference on Management of Data}}.
  \bibinfo{pages}{1642--1659}.
\newblock


\bibitem[Schleich et~al\mbox{.}(2016)]%
        {schleich2016learning}
\bibfield{author}{\bibinfo{person}{Maximilian Schleich}, \bibinfo{person}{Dan
  Olteanu}, {and} \bibinfo{person}{Radu Ciucanu}.}
  \bibinfo{year}{2016}\natexlab{}.
\newblock \showarticletitle{Learning linear regression models over factorized
  joins}. In \bibinfo{booktitle}{\emph{Proceedings of the 2016 International
  Conference on Management of Data}}. \bibinfo{pages}{3--18}.
\newblock


\bibitem[Shokri and Shmatikov(2015)]%
        {shokri2015privacy}
\bibfield{author}{\bibinfo{person}{Reza Shokri} {and} \bibinfo{person}{Vitaly
  Shmatikov}.} \bibinfo{year}{2015}\natexlab{}.
\newblock \showarticletitle{Privacy-preserving deep learning}. In
  \bibinfo{booktitle}{\emph{Proceedings of the 22nd ACM SIGSAC conference on
  computer and communications security}}. \bibinfo{pages}{1310--1321}.
\newblock


\bibitem[Tantipongpipat et~al\mbox{.}(2021)]%
        {tantipongpipat2021differentially}
\bibfield{author}{\bibinfo{person}{Uthaipon~Tao Tantipongpipat},
  \bibinfo{person}{Chris Waites}, \bibinfo{person}{Digvijay Boob},
  \bibinfo{person}{Amaresh~Ankit Siva}, {and} \bibinfo{person}{Rachel
  Cummings}.} \bibinfo{year}{2021}\natexlab{}.
\newblock \showarticletitle{Differentially private synthetic mixed-type data
  generation for unsupervised learning}.
\newblock \bibinfo{journal}{\emph{Intelligent Decision Technologies}}
  \bibinfo{volume}{15}, \bibinfo{number}{4} (\bibinfo{year}{2021}),
  \bibinfo{pages}{779--807}.
\newblock


\bibitem[Truex et~al\mbox{.}(2020)]%
        {truex2020ldp}
\bibfield{author}{\bibinfo{person}{Stacey Truex}, \bibinfo{person}{Ling Liu},
  \bibinfo{person}{Ka-Ho Chow}, \bibinfo{person}{Mehmet~Emre Gursoy}, {and}
  \bibinfo{person}{Wenqi Wei}.} \bibinfo{year}{2020}\natexlab{}.
\newblock \showarticletitle{LDP-Fed: Federated learning with local differential
  privacy}. In \bibinfo{booktitle}{\emph{Proceedings of the Third ACM
  International Workshop on Edge Systems, Analytics and Networking}}.
  \bibinfo{pages}{61--66}.
\newblock


\bibitem[Vafaie et~al\mbox{.}(1994)]%
        {vafaie1994feature}
\bibfield{author}{\bibinfo{person}{Haleh Vafaie}, \bibinfo{person}{Ibrahim~F
  Imam}, {et~al\mbox{.}}} \bibinfo{year}{1994}\natexlab{}.
\newblock \showarticletitle{Feature selection methods: genetic algorithms vs.
  greedy-like search}. In \bibinfo{booktitle}{\emph{Proceedings of the
  international conference on fuzzy and intelligent control systems}},
  Vol.~\bibinfo{volume}{51}. \bibinfo{pages}{28}.
\newblock


\bibitem[Vershynin(2018)]%
        {vershynin2018high}
\bibfield{author}{\bibinfo{person}{Roman Vershynin}.}
  \bibinfo{year}{2018}\natexlab{}.
\newblock \bibinfo{booktitle}{\emph{High-dimensional probability: An
  introduction with applications in data science}}. Vol.~\bibinfo{volume}{47}.
\newblock \bibinfo{publisher}{Cambridge university press}.
\newblock


\bibitem[Wang et~al\mbox{.}(2020)]%
        {wang2020hybrid}
\bibfield{author}{\bibinfo{person}{Chang Wang}, \bibinfo{person}{Jian Liang},
  \bibinfo{person}{Mingkai Huang}, \bibinfo{person}{Bing Bai},
  \bibinfo{person}{Kun Bai}, {and} \bibinfo{person}{Hao Li}.}
  \bibinfo{year}{2020}\natexlab{}.
\newblock \showarticletitle{Hybrid differentially private federated learning on
  vertically partitioned data}.
\newblock \bibinfo{journal}{\emph{arXiv preprint arXiv:2009.02763}}
  (\bibinfo{year}{2020}).
\newblock


\bibitem[Wang(2018)]%
        {wang2018revisiting}
\bibfield{author}{\bibinfo{person}{Yu-Xiang Wang}.}
  \bibinfo{year}{2018}\natexlab{}.
\newblock \showarticletitle{Revisiting differentially private linear
  regression: optimal and adaptive prediction \& estimation in unbounded
  domain}.
\newblock \bibinfo{journal}{\emph{arXiv preprint arXiv:1803.02596}}
  (\bibinfo{year}{2018}).
\newblock


\bibitem[Wei et~al\mbox{.}(2020)]%
        {wei2020federated}
\bibfield{author}{\bibinfo{person}{Kang Wei}, \bibinfo{person}{Jun Li},
  \bibinfo{person}{Ming Ding}, \bibinfo{person}{Chuan Ma},
  \bibinfo{person}{Howard~H Yang}, \bibinfo{person}{Farhad Farokhi},
  \bibinfo{person}{Shi Jin}, \bibinfo{person}{Tony~QS Quek}, {and}
  \bibinfo{person}{H~Vincent Poor}.} \bibinfo{year}{2020}\natexlab{}.
\newblock \showarticletitle{Federated learning with differential privacy:
  Algorithms and performance analysis}.
\newblock \bibinfo{journal}{\emph{IEEE Transactions on Information Forensics
  and Security}}  \bibinfo{volume}{15} (\bibinfo{year}{2020}),
  \bibinfo{pages}{3454--3469}.
\newblock


\bibitem[Wilson et~al\mbox{.}(2019)]%
        {wilson2019differentially}
\bibfield{author}{\bibinfo{person}{Royce~J Wilson},
  \bibinfo{person}{Celia~Yuxin Zhang}, \bibinfo{person}{William Lam},
  \bibinfo{person}{Damien Desfontaines}, \bibinfo{person}{Daniel
  Simmons-Marengo}, {and} \bibinfo{person}{Bryant Gipson}.}
  \bibinfo{year}{2019}\natexlab{}.
\newblock \showarticletitle{Differentially private SQL with bounded user
  contribution}.
\newblock \bibinfo{journal}{\emph{arXiv preprint arXiv:1909.01917}}
  (\bibinfo{year}{2019}).
\newblock


\bibitem[Wu et~al\mbox{.}(2017)]%
        {wu2017achieving}
\bibfield{author}{\bibinfo{person}{Genqiang Wu}, \bibinfo{person}{Xianyao Xia},
  {and} \bibinfo{person}{Yeping He}.} \bibinfo{year}{2017}\natexlab{}.
\newblock \showarticletitle{Achieving Dalenius' Goal of Data Privacy with
  Practical Assumptions}.
\newblock \bibinfo{journal}{\emph{arXiv preprint arXiv:1703.07474}}
  (\bibinfo{year}{2017}).
\newblock


\bibitem[Xu et~al\mbox{.}(2013)]%
        {xu2013differentially}
\bibfield{author}{\bibinfo{person}{Jia Xu}, \bibinfo{person}{Zhenjie Zhang},
  \bibinfo{person}{Xiaokui Xiao}, \bibinfo{person}{Yin Yang},
  \bibinfo{person}{Ge Yu}, {and} \bibinfo{person}{Marianne Winslett}.}
  \bibinfo{year}{2013}\natexlab{}.
\newblock \showarticletitle{Differentially private histogram publication}.
\newblock \bibinfo{journal}{\emph{The VLDB journal}}  \bibinfo{volume}{22}
  (\bibinfo{year}{2013}), \bibinfo{pages}{797--822}.
\newblock


\bibitem[Yang et~al\mbox{.}(2020)]%
        {yang2020local}
\bibfield{author}{\bibinfo{person}{Mengmeng Yang}, \bibinfo{person}{Lingjuan
  Lyu}, \bibinfo{person}{Jun Zhao}, \bibinfo{person}{Tianqing Zhu}, {and}
  \bibinfo{person}{Kwok-Yan Lam}.} \bibinfo{year}{2020}\natexlab{}.
\newblock \showarticletitle{Local differential privacy and its applications: A
  comprehensive survey}.
\newblock \bibinfo{journal}{\emph{arXiv preprint arXiv:2008.03686}}
  (\bibinfo{year}{2020}).
\newblock


\bibitem[Yoon et~al\mbox{.}(2019)]%
        {yoon2018pategan}
\bibfield{author}{\bibinfo{person}{Jinsung Yoon}, \bibinfo{person}{James
  Jordon}, {and} \bibinfo{person}{Mihaela van~der Schaar}.}
  \bibinfo{year}{2019}\natexlab{}.
\newblock \showarticletitle{{PATE}-{GAN}: Generating Synthetic Data with
  Differential Privacy Guarantees}. In \bibinfo{booktitle}{\emph{International
  Conference on Learning Representations}}.
\newblock
\urldef\tempurl%
\url{https://openreview.net/forum?id=S1zk9iRqF7}
\showURL{%
\tempurl}


\bibitem[Zhao et~al\mbox{.}(2020)]%
        {zhao2020local}
\bibfield{author}{\bibinfo{person}{Yang Zhao}, \bibinfo{person}{Jun Zhao},
  \bibinfo{person}{Mengmeng Yang}, \bibinfo{person}{Teng Wang},
  \bibinfo{person}{Ning Wang}, \bibinfo{person}{Lingjuan Lyu},
  \bibinfo{person}{Dusit Niyato}, {and} \bibinfo{person}{Kwok-Yan Lam}.}
  \bibinfo{year}{2020}\natexlab{}.
\newblock \showarticletitle{Local differential privacy-based federated learning
  for internet of things}.
\newblock \bibinfo{journal}{\emph{IEEE Internet of Things Journal}}
  \bibinfo{volume}{8}, \bibinfo{number}{11} (\bibinfo{year}{2020}),
  \bibinfo{pages}{8836--8853}.
\newblock


\end{thebibliography}
